\documentclass[acmsmall]{acmart}
 \newtheorem{definition}{Definition}
 \newtheorem{problem}{Problem}
  \newtheorem{observation}{Observation}
 
  \newtheorem{theorem}{Theorem}
   \newtheorem{remark} {Remark}
    
\newtheorem{property}[]{Property}

 \usepackage{tcolorbox}
\newenvironment{RQObservation}[1]
    {\renewcommand\theobservation{#1}\observation}
    {\endobservation}

    \usepackage{algpseudocode}
    \usepackage[inline]{enumitem}
\algnewcommand\algorithmicforeach{\textbf{for each}}
\algdef{S}[FOR]{ForEach}[1]{\algorithmicforeach\ #1\ \algorithmicdo}

\usepackage{amsmath,epsfig,epstopdf,amssymb}

\acmJournal{TOSEM}


\usepackage{tabularx}
\usepackage{float}
\usepackage{glossaries}
\usepackage{listings}
\floatstyle{boxed}
\newfloat{algorithm}{htbp}{loa}
\floatname{algorithm}{Algorithm}
\newfloat{algorithm}{htbp}{loa}

\title{Smart Contract Repair}

\author{Xiao Liang Yu}
  \affiliation{\institution{National University of Singapore}\country{Singapore}}
  \email{xiaoly@comp.nus.edu.sg}
  
\author {Omar Al-Bataineh}
  \affiliation{\institution{National University of Singapore}\country{Singapore}}
  \email{omerdep@yahoo.com}

\author{David Lo}
  \affiliation{\institution{Singapore Management University}\country{Singapore}}
  \email{davidlo@smu.edu.sg}

\author{Abhik Roychoudhury}
\authornote{Corresponding Author}
  \affiliation{\institution{National University of Singapore}\country{Singapore}}
  \email{abhik@comp.nus.edu.sg}



\newglossaryentry{sc}
{
        name={smart contract},
        description={E}
}
\newacronym{io}{IO}{integer overflow}
\newacronym{re}{RE}{reentrancy}
\newacronym{tod}{TOD}{transaction order dependence}
\newacronym{ed}{ED}{exception disorder}

\newglossaryentry{move}
{
        name={{\em move}},
        description={E}
}

\newglossaryentry{insert}
{
        name={{\em insert}},
        description={E}
}

\newglossaryentry{replace}
{
        name={{\em replace}},
        description={E}
}

\newcounter{RQ}

\definecolor{verylightgray}{rgb}{.97,.97,.97}

\lstdefinelanguage{Solidity}{
keywords=[1]{anonymous, assembly, assert, balance, break, call, callcode, case, catch, class, constant, continue, constructor, contract, debugger, default, delegatecall, delete, do, else, emit, event, experimental, export, external, false, finally, for, function, gas, if, implements, import, in, indexed, instanceof, interface, internal, is, length, library, log0, log1, log2, log3, log4, memory, modifier, new, payable, pragma, private, protected, public, pure, push, require, return, returns, revert, selfdestruct, send, solidity, storage, struct, suicide, super, switch, then, this, throw, transfer, true, try, typeof, using, value, view, while, with, addmod, ecrecover, keccak256, mulmod, ripemd160, sha256, sha3}, 
keywordstyle=[1]\color{blue}\bfseries,
keywords=[2]{address, bool, byte, bytes, bytes1, bytes2, bytes3, bytes4, bytes5, bytes6, bytes7, bytes8, bytes9, bytes10, bytes11, bytes12, bytes13, bytes14, bytes15, bytes16, bytes17, bytes18, bytes19, bytes20, bytes21, bytes22, bytes23, bytes24, bytes25, bytes26, bytes27, bytes28, bytes29, bytes30, bytes31, bytes32, enum, int, int8, int16, int24, int32, int40, int48, int56, int64, int72, int80, int88, int96, int104, int112, int120, int128, int136, int144, int152, int160, int168, int176, int184, int192, int200, int208, int216, int224, int232, int240, int248, int256, mapping, string, uint, uint8, uint16, uint24, uint32, uint40, uint48, uint56, uint64, uint72, uint80, uint88, uint96, uint104, uint112, uint120, uint128, uint136, uint144, uint152, uint160, uint168, uint176, uint184, uint192, uint200, uint208, uint216, uint224, uint232, uint240, uint248, uint256, var, void, ether, finney, szabo, wei, days, hours, minutes, seconds, weeks, years}, 
keywordstyle=[2]\color{teal}\bfseries,
keywords=[3]{block, blockhash, coinbase, difficulty, gaslimit, number, timestamp, msg, data, gas, sender, sig, value, now, tx, gasprice, origin}, 
keywordstyle=[3]\color{violet}\bfseries,
identifierstyle=\color{black},
sensitive=false,
comment=[l]{//},
morecomment=[s]{/*}{*/},
commentstyle=\color{gray}\ttfamily,
stringstyle=\color{red}\ttfamily,
morestring=[b]',
morestring=[b]"
}

\begin{document}

\begin{abstract}

Smart contracts are automated or self-enforcing contracts
that can be used to exchange assets
without  having to place trust in third parties. Many commercial transactions use smart contracts due to their 
potential benefits in terms of secure peer-to-peer transactions independent of external parties.
Experience shows that many commonly used smart contracts are vulnerable 
to serious malicious attacks which may enable attackers to steal valuable assets of involving parties.
There is therefore a need to apply analysis and automated repair techniques
to detect and repair bugs in smart contracts before being deployed. 
In this work, we present the first general-purpose automated smart contract repair approach that is also gas-aware. Our repair method is search-based and searches among mutations of the buggy contract. Our method also considers the gas usage of the candidate patches by leveraging our novel notion of \emph{gas dominance relationship}. We have made our smart contract repair tool \textsc{SCRepair} available open-source, for investigation by the wider community.
\end{abstract}

\widowpenalty 10000
\clubpenalty 10000

\maketitle

\section{Introduction}

Smart contracts are automated or self-enforcing programs 
which currently underpin many online commercial transactions.
A smart contract is a series of instructions or operations written in
special programming languages which get executed when certain conditions are met.
Typically, smart contracts are running on the top of  blockchain systems, which 
are distributed systems whose storage
is represented as a sequence of blocks.
The key attractive property of smart contracts is mainly related to their ability to eliminate the need of trusted third parties in multiparty interactions, enabling parties to engage in secure peer-to-peer transactions without having to place trust in external parties (i.e., outside parties which help to fulfill the contractual obligations).  


While smart contracts are commonly used for commercial transactions, many malicious attacks in the past were made possible due to poorly written or vulnerable smart contracts. 
The code executed by smart contracts can be complex. There is therefore a need for testing (e.g. \cite{Jiang2018,Luu2016}), analysis (e.g. \cite{KalraGDS18}) and verification (e.g. \cite{Tsankov2018}) of smart contracts. In this paper, we take the technology for enhancing reliability of contracts one step further: once vulnerabilities in smart contracts are detected, we seek to automatically {\em repair} the vulnerabilities.

Automated program repair \cite{GOUES19} is an emerging technology for automatically fixing errors and vulnerabilities in programs via search, symbolic analysis, program synthesis and learning. The successful application of automated repair techniques to traditional programs  \cite{LeGoues2012,GouesNFW12,Martinez2015,Nguyen2013,Xuan2017,Mechtaev2015,Long2015,Mechtaev2016} raises the question of whether
 these techniques can be also applied to fix bugs in smart contracts. 
Several different approaches have been developed
to automatically repair bugs in traditional programs
which can be classified mainly into two categories:
heuristic repair approaches \cite{LeGoues2012,GouesNFW12,Martinez2015} and constraint-based repair approaches  \cite{Nguyen2013,Xuan2017,Mechtaev2015,Long2015,Mechtaev2016}.
The inputs to these approaches are a buggy program and a correctness criterion 
(often given as a test suite). 
The automated repair approaches  return a (often minimal) transformation of the buggy program, so that the transformed program passes all the tests in the given test-suite.

In practice, the implications of unfixed bugs in smart contracts can be more serious than the typical non-security sensitive programs for several reasons. 
First, smart contracts are open for inspection and running on a decentralized network, the whole program state of smart contracts is transparent to everyone.
Second, the generated patch for a vulnerable smart contract should not only
fix the detected vulnerabilities but also needs to be 
mindful of the gas consumption of the resultant patched program.
The blockchain system on which the contract will be running typically has a gas usage limit.
Third, the quality of the generated patch for a vulnerable smart contract is a major design issue to be considered 
as smart contracts are typically used for commercial transactions. In fact, malicious agents may take advantage
of unfixed bugs in smart contracts to steal some
valuable assets of the  parties involved.

In this work, we develop an automated smart contract repair algorithm  using genetic programming search. Given a vulnerable smart contract and test suite, we conduct a parallel, biased random search for a set of edits to the contract that fixes 
a given vulnerability without breaking any test that previously passed. The bias in the search comes from the objective function driving the search. The parallelization strategy consists of splitting the search space into mutually-exclusive (disjoint) sub-spaces, where patches in each sub-space are concurrently and independently generated and validated.
We introduce also the notion of gas dominance level for smart contracts which enables us to compare the quality of patches based  on their runtime gas.
The gas dominance level 
can be used to compare the quality of generated patches. This also emphasizes our position is that automated repair of smart contracts needs to be gas-aware.
 
 To evaluate the effectiveness of our genetic repair algorithm, we constructed a dataset of vulnerable smart contracts taken from the Ethereum mainnet network, which is the main network wherein actual transactions of smart contracts take place on a distributed ledger. Hence, our constructed dataset consists
of real-world smart contracts.
During our evaluation, we considered 20 vulnerable contracts
 which have been selected randomly from the constructed
  dataset while taking into
 consideration the class of detected vulnerabilities
 and the complexity of the vulnerable contracts.
 The vulnerable contracts have been selected in a way such that  most of the common classes of vulnerabilities that are typically made by smart contract developers are covered when evaluating the genetic algorithm. 
However, to understand
 and draw some valid conclusions
 about the factors affecting 
 the correctness and quality of 
 patches generated by
 the algorithm,
 we have evaluated the algorithm under many different  settings and configurations.
 Examples of such settings include:
 (i) enabling/disabling the gas calculation
 of generated patches, (ii) varying
 the size of time budget allocated to the algorithm
.
Our genetic  algorithm
 was able to fully repair 10 vulnerable smart contracts from the selected set of 20 vulnerable contracts,  achieving a $50\%$ success rate. It is interesting to mention that most of the selected vulnerable contracts have multiple bugs and we therefore assert a vulnerable contract as repaired
 if all detected bugs are repaired.

 \paragraph*{Contributions}
 We summarize our main contributions as follows.

\begin{itemize}

\item   We present the first automated smart contract repair approach
that is gas-optimized and vulnerability-agnostic. The approach is inspired by genetic programming and can be used to generate a patch for a given vulnerable smart contract.

\item We describe a parallel genetic repair algorithm
 that can be used to split the large search space of candidate patches into smaller mutually-exclusive search spaces which can be processed independently. The presented parallel algorithm helps to process large number of candidate patches in a short computational time and therefore, in contrast to previous repair approaches, repairs can be generated  faster.
It also improves the scalability
of genetic repair algorithms so that large real-world contracts
can be repaired.

 \item We show how to integrate gas-awareness into the repair of smart contracts. This is crucial for smart contracts
 as excessive unnecessary gas consumption of contracts can lead to financial loss or
 out-of-gas exceptions when running the contract on a public
 blockchain network. It is therefore necessary to reduce the cost of running the contract and also
 the possibility of introducing new out-of-gas exceptions
 when repairing a vulnerable smart contract.
 We introduce a simple yet effective gas ranking approach with the novel notion of \emph{Gas Dominance Level}
 that can be used to rank generated patches of a given
 vulnerable smart contract during the patch generation.
 In general, the gas consumption of
 a given smart contract can be a non-constant bound
 which can be described as a parametric gas formula
 that takes into consideration both static and dynamic
 parameters that affect the cost of the contract
 including the instruction gas, memory gas, stack gas, and storage gas. We provide an acceleration technique to quickly compare candidate patches in terms of gas consumption, by introducing the concept of {\em gas dominance levels}.
 

\item Based on the above described techniques, 
we develop a fully automated repairing tool 
for smart contracts (which we call \textsc{ SCRepair}) which is integrated with a gas ranking approach
to generate an gas-optimized secure contract. 
Our tool can both detect and repair  security vulnerabilities in smart contracts. It is does so by integrating the tool \textsc{SCRepair} with the powerful 
smart contract security analyzer Oyente \cite{Luu2016} and Slither \cite{slither}.
We demonstrate that our approach is effective in fixing bugs 
 for real-world smart contracts.
  Our approach can deal with bugs whose fixes involve multi-line changes. Our smart contract repair tool and dataset is publicly available in GitHub from
  {\bf \url{https://SCRepair-APR.github.io}}    
\end{itemize}

\section{BlockChain and Smart Contracts}
The blockchain  technology is a distributed database 
that maintains records and transactions in a decentralized fashion. It has been
adopted in many applications to increase security and
reliability and to avoid the need for a trusted third party.
The transactions on blockchain are available to all the parties in the network in real-time and all the parties are allowed to interact with each other in a distributed manner. It uses state-of-the-art cryptography,  and hence it enables parties to engage in secure peer-to-peer transactions. The decentralized nature of blockchain makes it  suitable for many applications including decentralized cloud storage with provenance\cite{ethdrive},  general health care management, IoT data sharing with assured integrity\cite{liu2017blockchain}, and general commercial transactions.

Smart contracts are one of the most successful applications of the blockchain technology. They currently underpin many online commercial transactions which are typically running on the top of blockchain systems.
A smart contract is a special computer program whose executions are done in a decentralized and tamper-proof manner.
The key attractive property of smart contracts is mainly related to their ability to eliminate the need of trusted third parties in multiparty interactions.
Smart contracts allow for decentralized automation by facilitating, verifying, and enforcing the conditions of an underlying agreement.

Ethereum is the most popular blockchain platform supporting smart contracts. It supports a feature called Turing-completeness that allows the creation of practically useful smart contracts.  Smart contracts are typically written using the programming language ``Solidity''. 
Note that everything executed on Ethereum 
costs some gas for giving the miners incentive to perform the computations \cite{wood2014ethereum}. For example, executing an ADD instruction costs
3 units of gas. Storing a byte costs 4 units or 68 units of gas, depending on the
value of the byte (zero or non-zero). 
Hence, any slight mutation 
to the source code of a smart contract can change the gas usage of the contract tremendously (and hence
the amount of money that the parties of a transaction need to pay when running the smart contract on a real blockchain network).

To develop a better understanding of blockchain and smart contracts, 
let us consider  an example. Suppose that Bob would like to sell
a property (house) to Alice and Alice is willing to pay 100 Ether (a cryptocurrency) as a price for that
property and that Bob is happy with Alice's offer. 
After some discussion, they agreed to proceed with their business transaction
 and wish to perform it in an automated way by taking advantage of blockchain and smart contracts.
From the given description of the problem, one 
can see that there are three main conditions that any possible solution to the problem needs to satisfy:
(1) Bob has legal ownership of the property that he is selling (2) Alice can get the ownership of Bob's property 
only if she transferred 100 Ether to Bob, 
and (3) Bob can get 100 Ether from Alice only if he transferred the ownership of his property to Alice.
The transaction can be said to be successful if upon completion, the ownership of Bob's property is transferred to Alice while Bob receives 100 Ether.

\begin{figure}

\begin{lstlisting}[language=Solidity]
contract CommercialTransaction {
    bool transferredA = false; bool transferredB = false;
    
    function transferA() public payable {
        if(msg.sender == Alice && msg.value == 100 ether) {
            transferredA = true;
        }
    }
    function transferB() public {
        // Bob should transfer the house ownership to this
        // contract before calling this function
        if(msg.sender == Bob && hasHouseOwnership(address(this))) {
            transferredB = true;
        }
    }
    function finalize() public {
        if(transferredA && transferredB) {
            transferredA = false; transferredB = false;
            transferHouseOwnership(Alice); Bob.transfer(100 ether);
        }
    }
    function abortA() public {
        if(msg.sender == Alice && transferredA) {
            transferredA = false; Alice.transfer(100 ether);
        }
    }
    function abortB() public {
        if(msg.sender == Bob && transferredB) {
            transferredB = false; transferHouseOwnership(Bob);
        }
    }
}
\end{lstlisting}
\caption{A smart contract written in Solidity language that allows two parties to be involved in a commercial transaction to sell some property. \lstinline[language=Solidity]{msg.sender} represents the party calling the function. Certain definitions are omitted for brevity.}
\label{fig:AliceBobEx}

\end{figure}


Suppose that Alice and Bob perform their  transaction 
using the smart contract given in Fig. \ref{fig:AliceBobEx} written in the most popular smart contract programming language Solidity.
The code consists of a number of functions needed in order to perform the 
commercial transaction in an atomic way. 
The function \verb+transferA+ is used by Alice to 
transfer 100 Ether to the smart contract 
and hence when this function is executed, the money comes under the control of the smart contract.
The function \verb+transferB+ is used by Bob to inform the smart contract that the ownership of his property has been transferred to the smart contract.
So that after executing the functions  \verb+transferA+ and  \verb+transferB+, 
the smart contract is supposed to hold both the money of Alice and the property of Bob. The function \verb+finalize+ is used to
finalize the transaction by transferring the money
from Alice to Bob and the property from Bob to Alice.
The smart contract provides also two more functions,
namely \verb+abortA+ and \verb+abortB+
which are available to both Alice and Bob respectively.
The goal of these functions is to protect the parties from the situation where the purchase is canceled halfway while the smart contract has already held the assets from the parties, so that the parties can get back their assets from this contract.

Recently, there has been a growing interest in  verification and validation of smart contracts \cite{Grossman2017,KalraGDS18,Jiang2018,Amani2018,Meyden2019},
as vulnerabilities in smart contracts can have serious adverse consequences.
Therefore, a number of vulnerability detection tools have been developed for smart contracts including Oyente \cite{Luu2016}, Slither\cite{slither}, and ContractFuzzer \cite{Jiang2018}. 
In general, smart contract vulnerabilities can be categorized into
three categories \cite{Atzei2017}: (i) vulnerabilities at the blockchain level, (ii) vulnerabilities at the Ethereum virtual machine level, and (iii) vulnerabilities at the source code level.
In this work, we are interested on the vulnerabilities that can be repaired at the level of source code.

\begin{table}[h]
\caption{Selected smart contract vulnerabilities that can be fixed by modifying Solidity source code}
\begin{center}
\begin{tabular}{ |l|l|}
\hline
\textbf{Class of vulnerability} & \textbf{ References} \\ \hline
Exception disorders / Mishandled exceptions / Gasless send    &  \cite{Atzei2017,Dika2019,Luu2016,Bhargavan2016, Tikhomirov2018,Jiang2018,Tsankov2018}     \\ \hline
Reentrancy  &  \cite{Atzei2017,Dika2019,Luu2016, Bhargavan2016,Tikhomirov2018,Jiang2018,Tsankov2018}           \\ \hline
Integer overflow / Integer underflow / Unchecked math &    \cite{Dika2019,Tsankov2018,Tikhomirov2018,Luu2016}   \\ \hline
Transaction order dependence / Unpredictable state   &   \cite{Atzei2017,Dika2019,Luu2016,Tsankov2018}   \\ \hline
\end{tabular}
\end{center}
\label{table:vulnerability}

\end{table}
Based on our conducted literature review on recent research work on smart contracts \cite{Atzei2017, Jiang2018, Luu2016, Tsankov2018, Delmolino2016, Bhargavan2016, Dika2019, Tikhomirov2018, GrishchenkoMS18}, we summarize in \autoref{table:vulnerability}  some selected  popular vulnerabilities that can be detected using the tools Oyente \cite{Luu2016} and Slither \cite{slither}.
Table \ref{table:vulnerability} shows  
a summary of these widely studied vulnerabilities.
We give a detailed description of these classes
of vulnerabilities in \autoref{sec: implementation}.

\section{The Smart Contract Repair Problem}

Recent advances in program repair techniques \cite{GOUES19} have raised  the possibility of developing program repair technology for smart contracts.
In this section, we discuss the  automated smart contract repair problem together with the set of challenges that might be encountered in repairing smart contracts.
We also discuss the key differences between the smart contract repair problem
and the traditional program repair problem.
 
\begin{problem} (\textbf{Automated smart contract repair problem}). \label{problem:ascr}
Consider a vulnerable smart contract $C$ with a set of detected vulnerabilities $U$, a test suite $T$ and a maximum gas usage bound $L$, the automated smart contract repair problem
 is the problem of developing an algorithm that takes as input 
$(C, U, T, L)$ and produces as an output a new contract $C^{'}$
that is similar to $C$ but has all vulnerabilities in $U$ fixed, passing all tests in $T$, and the maximum gas usage of feasible execution paths should be less than or equal to $L$.
\end{problem}

The smart contract repair problem
is similar to the traditional program repair problem.
However, the smart contract repair problem introduces some extra computational complexity 
as the patch generation needs to be gas-aware.
It is also highly desirable for the patches to signify readable and small changes, so that the patched contract is easily comprehensible. Overall, we would want the  
the (syntactic) structure of the vulnerable contract to be maximally preserved.


Since detailed formal specifications of intended program behavior are typically unavailable, program repair  uses weak correctness criteria, 
such as an assertion of existence of vulnerabilities by vulnerability detector and a test suite. 
Therefore, the validity of patches is relative  to the chosen vulnerability detector and the available test cases.

As mentioned earlier, the generated patches for smart contracts need to meet more criteria than those generated  for traditional programs. This is  mainly due to the fact that smart contracts are typically running on the top of the blockchain systems, which impose certain constraints on the total computational resources used by the contract. The execution of the smart contract
needs to comply with the gas usage constraints imposed by the blockchain system.
Note that if the running smart contract exceeds
the allowed upper bound limit of the gas usage, the execution of the contract will be interrupted and a ``out-of-gas'' exception will be thrown.

\begin{definition} \label{ValidityDef} 
(\textbf{Validity criteria of generated patches}).
Given a vulnerable smart contract $C$ with a set of detected vulnerabilities $U$ and a test suite $T$ that consists of  two sets:
the failing tests $T_{F}$  
and the passing tests $T_{P}$. 
Suppose that the contract $C$ is running on the top of 
a blockchain system $B$ and that the maximum allowed
gas usage available to the contract is bounded by $L$. 
We say that the new patched smart contract $C^{'}$
is a valid plausibly fixed contract if it satisfies the following requirements.

\begin{enumerate}

\item The  contract $C^{'}$ is not vulnerable to the vulnerabilities in $U$.
\item  The  contract $C^{'}$ passes all tests in $T_{F}$.

\item The  contract $C^{'}$ does not break any test in $T_{P}$.

\item There is no feasible execution path in $C^{'}$ 
whose total gas consumption exceeds the bound $L$.

\end{enumerate}

\end{definition}

Typically, the bound $L$ imposed
on the gas usage of the contract is determined
by the involving  parties of the transaction,
the structure and semantics of the smart contract,
and the block gas limit of the blockchain.
Such bound (if known) can be incorporated
in the patch generation process for vulnerable contracts
in order to avoid introducing new out-of-gas exceptions. 
Note that requirement 4 of Definition \ref{ValidityDef} can be checked by enumerating all feasible paths
in the patched contract $C^{'}$ and then verifying that there is no feasible path that exceeds the bound $L$.

In addition to the above correctness requirements, we are also interested in certain desirable properties indicating 
patch quality, as described in the following.


\begin{enumerate}

\item \textbf{The simplicity of the patch}. The simplicity of the edited
contract can be measured in terms
of the number of edits that have been made to the original contract.

\item \textbf{The cost of the patch}. The cost of the contract can be measured in  different ways. We choose here the average
gas usage as a metric to measure the cost of the contract.

\end{enumerate}

To evaluate the quality requirements of a
generated patch we introduce two functions, namely $\mathit{diff} (C, C^{'})$ and $cost (C^{'})$.
The function $\mathit{diff} (C, C^{'})$ returns a numerical value
that specifies how much the edited contract $C^{'}$
differs from the original vulnerable contract $C$.
Replacing expressions, inserting of new statements,
and moving/deleting of statements will be counted
when computing $\mathit{diff} (C, C^{'})$. Overall, $\mathit{diff}(C,C')$ captures the edit distance between two smart contracts $C$ and $C'$.

The function $cost(C)$ computes the average cost of
gas usage of a given smart contract. 
Recall that every single operation that takes part in the blockchain network consumes some amount of gas. Gas is what is used to calculate the fee that need to be paid to the miner in order to execute operations.
Of course, the cost of transactions can vary from one to the other depending on the details of the transaction and the structure and complexity of the smart contract. 
However, for a given smart contract $C$ and a specific transaction $t$,
one can perform certain calculations to compute the average cost or the maximum expected cost of the transaction in gas units, provided that the cost of each operation
of the contract on the running blockchain system is known in advance.
We defer the discussion of the computational details of gas usage
of a given smart contract to Section \ref{GasAnalysis}.

\paragraph*{On Plausible and Correct Patches}
In this paper, we use the terminology of {\em plausible patch} and {\em correct patch}. Here we rely on the terminology in program repair literature (e.g. see \cite{GOUES19}), where a correct patch is deemed to be correct via manual analysis, but a plausible patch is one produced by a repair technique since it passes all given tests. Since a formal complete specification of the intended program behavior is not available, the description of intended behavior given to a program repair technique is incomplete: it is given in the form of tests, assertions or vulnerabilities found. A plausible patch generated by the repair algorithm thus meets the intended behavior as per this incomplete description that was provided to the repair method. Thus, if the repair method was given a test-suite T, a plausible patch can still potentially fail a test t outside T. For this reason, a plausible patch cannot be guaranteed to be correct, and we need a manual validation step to ascertain how many of the plausible patches generated are correct. We have conducted such an evaluation in our work, in Section \ref{sec:RQ}.

\section{The Smart Contract Repair Framework}

In this section, we present a multi-objective genetic repair algorithm
with mainly four objectives: two objectives related to
the correctness of the smart contract and two related
to the quality of the generated patch.
We develop an efficient genetic search approach to generate a patch 
for a vulnerable smart contract. Our proposed genetic search technique employs
mutation operators
to generate fix candidates for the vulnerable contract
and then uses fitness functions 
to evaluate the suitability of the candidate patch.
The overall goal of our approach is to generate correct, 
high-quality, and  gas-optimized fixes for the vulnerable smart contract. 

\paragraph*{Advantages of our search-based approach}
The main motivation behind developing a genetic repair
approach relies on the hypothesis that most software bugs
  introduced by  programmers are due to small syntactic errors. Furthermore, the genetic search technique also has the following advantages with respect other common repair techniques. 
  \begin{itemize}
      \item Semantic repair techniques employ symbolic execution and program synthesis for repairing programs. Employing such techniques for smart contract repair will deprive our approach of the natural ability to insert/delete statements which seems to be important for repairing common smart contract vulnerabilities like the reentrancy vulnerability.
      \item Template based repair techniques can be used as a purely static approach to smart contract repair. In this approach for every detected vulnerability type, a specific program transformation template can be employed for repair. Such an approach deprives us the possibility of exploring a variety of patch candidates and enforce patch quality indicators in terms of gas consumption and patch simplicity.
  \end{itemize}

 \subsection{Mutation Analysis of Smart Contracts}
 
 The mutation analysis of a vulnerable smart contract
 is the process through which a set of contract variants, called mutants, are generated by seeding a large number of small syntactic changes into the vulnerable contract using some mutation operators. The mutants are considered as patch candidates and we use the various correctness criteria and patch quality criteria, to choose and prioritize from the pool of patch candidates.

 \subsection{Mutation Operators and Patch representation}

We employ three mutation operators. The \gls{move} operator moves a given statement in the analyzed smart contract to some other location in the contract. 
The \gls{insert} operator inserts a randomly synthesized statement before or after a given buggy statement. The \gls{replace} operator replaces a potentially-buggy expression with another randomly synthesized expression. Our set of mutation operators contains both statement-level and expression-level operators to allow efficient mutation conducted at different granularity.

\paragraph*{Patch Representation}
A patch candidate is represented in terms of the mutation operations that need to be performed on the abstract syntax tree of the original vulnerable contract $C$ being repaired.

 \subsection{Generating Mutated Smart Contracts} \label{sec:Algorithm}
 
A large number of mutants
may be introduced when repairing a vulnerable smart contract
 depending on the size of the contract,
leading to searching among an extremely large set of mutants. Note that the validation process  of the generated mutants can be extremely costly and time-consuming as also shown by other works on automated program repair \cite{GouesNFW12}. Each mutant may need 
 to be detected using the vulnerability detectors and tested against the original test suite.
 It is therefore necessary to apply a parallelization methodology in order to speed up the validation process
  of candidate mutants for a given vulnerable contract.
 
 All mutation operators used in our repair framework can affect the cost of the vulnerable smart contract $C$ which is also confirmed by our experiments.
Their effect on the cost of the contract
 can be considerable, especially
 when the vulnerable contract contains loops that can be repeated a large number of times. 
 If a plausible patch of $C$ is obtained
 by replacing or inserting a statement within the body of the loop then the cost of the contract may change dramatically. It is crucial to search for a gas-optimized patch when repairing smart contracts in order to minimize the possibility of introducing new out-of-gas exceptions to the smart contract being repaired.

 In general, generating a gas-optimized repair
 for a given vulnerable smart contract can be a computationally complex task.
 Note that the repair should not only fix the vulnerability in the contract
 but also needs to not increase the gas usage significantly.
 To achieve such a goal, 
 one might choose to mutate the vulnerable smart contract
 $C$ by favoring the mutation operators \gls{move} and \gls{replace} over the mutation operator \gls{insert}
  when searching for low-cost patches. 
  Indeed, intuitively, when we add new instructions onto the program would likely to increase the computational demand.
 Unfortunately, such a simple prioritization strategy 
 does not necessarily lead
 to the least costly plausible patch for the vulnerable contract. Subtle interactions between the operators can turn a low-cost
  contract into a high-cost contract and vice versa.
  For example, the \gls{insert} mutation operator which supposes
  to increase the cost of the contract by adding a new
  statement, may sometimes lead to a mutant with lower
  gas usage than the original vulnerable contract. 
  Similarly, the \gls{move} mutation operator which supposes
  not to increase the cost of the contract
   can also lead to a mutant whose gas usage is higher than
  that of the original contact.
  The cost of the generated mutant does not depend only on 
  the cost of the applied mutation operations
  but also on the way the operators change 
  the behavior of the contract.
   We therefore cannot favor one operator over another when searching 
    for low cost repairs without performing some analysis 
    on the overall structure of the vulnerable contract.

  Let us consider some trivial examples to demonstrate
  how the \gls{insert} operator can turn a high-cost contract into a low-cost contract
  while the \gls{move} operator may turn a low-cost contract into a high-cost contract.
  The program in Fig. \ref{fig:Ex1} 
  represents a buggy  program. 
  Suppose that we generate a mutant for this program
  by inserting a new statement after the initialization statement
  (line 1) of the form: \lstinline[language=Solidity]{a = false;} .
  In this case, the loop in the generated mutant will be skipped
  and  the average gas usage of the new mutated version
  will be much smaller than that of the original version.
  The program in Fig. \ref{fig:Ex2} represents another potentially buggy program.
  Let us generate a random mutant 
  of the program by applying the \gls{move} operator 
  so that the statement at line 4 (the loop counter update statement)
  is moved outside the loop. Obviously, this will turn
  the loop into an infinite loop and hence the contract will
  run out of gas after certain number of iterations.
  Note that since mutation makes random changes to
   the buggy smart contract, it
   may impact the performance and cost of the contract 
   in many different arbitrary ways. This is critical
   especially when the buggy smart contract contains loops.

\begin{figure}
\centering
\begin{minipage}{\dimexpr 0.55\textwidth -0.2\columnsep}
  \centering
    \begin{lstlisting}[language=Solidity]
bool a = true;
while (a) {
    // Some computation
}
    \end{lstlisting}
   \caption{Inserting \lstinline[language=Solidity]{a = false} after line 1 reduces gas}
  \label{fig:Ex1}
\end{minipage}%
\begin{minipage}{\dimexpr 0.45\textwidth -0.2\columnsep}
  \centering
    \begin{lstlisting}[language=Solidity]
int x = 0;
while (x <= 100) {
    x = x + 2;
    // Some computation
}
    \end{lstlisting}
   \caption{Moving line 3 outside loop increases gas}
  \label{fig:Ex2}
\end{minipage}
\end{figure}

\begin{observation}
There is insufficient information to predict the gas of a mutated contract by inspecting the mutation operations applied. For example, the successive applications of the mutation operators not introducing new statements
\emph{(\gls{move}, \gls{replace})} does not necessarily lead
to a low-cost mutant w.r.t. the original smart contract. Similarly, 
the successive applications of the mutation operator inserting new statements
\gls{insert} does not necessarily lead to a high-cost mutant
w.r.t. the original smart contract. 
The cost of the generated mutants depends mainly
on how the applied mutation operators 
change the behavior of the smart contract.

\end{observation}

As mentioned earlier, one of the biggest challenges that need to be
addressed when using a genetic search approach for repairing smart contracts
is how to speed up the generation and validation processes of mutated versions.
We describe here a parallel search-based algorithm for efficiently generating patches.
We assume here we have three versions 
of the mutate function: $mutateM(C)$
which mutates the contract $C$ using only the \gls{move} operator,
$mutateR(C)$ which mutates the contract $C$ using only the \gls{replace} operator, and $mutateI(C)$
which mutates the contract $C$ using only the \gls{insert} operator.
 Since genetic repair approaches use mainly an exhaustive search algorithm
to generate a patch, it is highly desirable
to split the search space into sub-spaces. 
To do so, we use the mutate functions described above 
to split the search space
into 7 smaller spaces as described in the following.

\begin{itemize}

\item $[Space S_1]$: this search space consists of
the set of candidate patches that result from mutating the contract $C$ using only the function  $mutateM(C)$.

\item $[Space S_2]$: this  search space consists of
the set of candidate patches that result from mutating the contract $C$ using only the function $mutateR(C)$.

\item $[Space S_3]$: this search space consists of
the set of candidate patches that result from mutating the contract $C$ using only the function $mutateI(C)$.

\item $[Space S_4]$: this search space consists of
the set of candidate patches that result from mutating the contract $C$ using  the two functions $mutateM(C)$ and $mutateR(C)$.

\item $[Space S_5]$: this search space consists of
the set of candidate patches that result from mutating the contract $C$ using  the two functions $mutateM(C)$ and $mutateI(C)$.

\item $[Space S_6]$: this search space consists of
the set of candidate patches that result from mutating the contract $C$ using  the two functions $mutateR(C)$ and $mutateI(C)$.

\item $[Space S_7]$: this search space consists of
the set of candidate patches that result from mutating the contract $C$ using the functions $mutateM(C)$, $mutateR(C)$,  $mutateI(C)$.

\end{itemize}

Note that for the effectiveness of the parallel
algorithm we need to ensure that the search spaces
are mutually-exclusive spaces so that no redundant 
mutants are generated and validated across various spaces.
Recall that each mutant will be checked using the vulnerability detectors and against 
a set of test cases in addition to the gas usage requirement.
Such validation process can be computationally complex
specially when the search space of candidate patches is extremely large.

Mutants in $S_7$ are generated using
the nesting operation $\mathit{mutateX}(\mathit{mutateY} (\mathit{mutateZ} (C)))$, 
where $X, Y$, and $Z$ are distinct operators taken from the mutation domain
 $\{Move, Replace, Insert\}$. 
Assume $C_1 = \mathit{mutateZ} (C),  C_2 = \mathit{mutateY} (\mathit{mutateZ} (C))$, $C_3 = \mathit{mutateX}(\mathit{mutateY} (\mathit{mutateZ} (C)))$.
Then the validity function $V_{S_7} (C_3)$ for this search space $S_7$
 can be formalized as follows.
$$ \label{EqS7}
V_{S_7} (C_3)  = 
\begin{cases}
Accept &   \textrm{iff $ \mathit{diff}(C_1, C) >0 \land \mathit{diff} (C_2, C_1) > 0 $} \\ 
& \hspace*{10 pt}\textrm{$\land~ \mathit{diff} (C_2, C) > 0 \land  \mathit{diff} (C_3, C_2) > 0$} 
\\ 
& \hspace*{10 pt}\textrm{$\land~ \mathit{diff} (C_3, C) > 0 \land  \mathit{diff} (C_3, C_1) > 0$} 
\\
Reject & otherwise
\end{cases}
$$

Note that for a mutant to be added to the space $S_7$
it has to satisfy a somewhat complex condition.
This is necessary in order to avoid overlaps
with the other search spaces. 	
Similar validity functions are defined for 
the other sub-spaces to ensure the mutually-exclusive property of the sub-spaces (please see Theorem \ref{MXTheorem}).

\begin{definition} (\textbf{Properties of  splitting strategy}).
Let $S$ be the search space of possible mutants of a vulnerable smart 
contract $C$ generated using the operators \gls{move},
\gls{replace}, and \gls{insert}. 
The splitting strategy of $S$ into 
spaces $S_1,..., S_7$ satisfies the following properties

\begin{itemize}

\item disjointness: for any two distinct sets $S_i$ and $S_j$
such that ($i, j = 1,..., 7 \land i \neq j$) we have $S_i \cap S_j = \emptyset$.

\item completeness: $(S_1 \cup S_2 \cup ... \cup S_7) = S $.

\end{itemize}

\end{definition}

\begin{theorem} \label{MXTheorem}  Spaces $(S_1,..., S_7$) are mutually exclusive spaces.

\end{theorem}

\begin{proof} (sketched).
To prove the theorem we need to consider many different cases as we have 7 spaces.
However, since the proof argument of all cases will be very similar
and for  brevity reason, we consider here only space $S_7$.
For this case, we need to show that $S_7 \cap S_j = \emptyset \mid j = 1...6$.
Hence, there are six possible sub-cases to consider.
Recall that the mutants in $S_7$ are generated using
the nesting operation $\mathit{mutateX}(\mathit{mutateY} (\mathit{mutateZ} (C)))$, where $X, Y$, and $Z$ are distinct operators taken from the mutation domain $\{Move, Replace, Insert\}$. 
The theorem can be proven by contradiction. 
\begin{itemize}

\item Let $S_i \cap S_7 \neq \emptyset \mid i \in \{1, 2,3\}$.
This implies that there exists a mutant $m$ that belongs to both $S_i$ and $S_7$.
Note that since $m$ belongs to $S_i$ then it is generated using a single mutate function of the form $mutateX$, where $X \in \{Move, Replace, Insert\}$.
It is easy to see then that the mutant $m$ cannot exist in the space $S_7$ as the addition of such mutant to $S_7$
contradicts with the definition of the validity function of the space $S_7$.

\item  Let $S_j \cap S_7 \neq \emptyset \mid j \in \{4, 5, 6\}$.
This implies that there exists a common mutant $m$ that belongs to both $S_j$ and $S_7$. Note that since $m$ belongs to $S_j$ then $m$ is generated 
from the nesting operation $mutateX (mutateY (C))$, where $X$ and $Y$
are distinct operators taken from the domain $\{Move, Replace, Insert\}$. Hence,
the mutant $m$ is generated using only two operators while ignoring
the effect of one of the three operators. Therefore,
the  mutant $m$ cannot exist in the space $S_7$
as this contradicts with the definition of the validity function of the space $S_7$
and the fact that mutants in $S_7$ are generated using the nesting operation $mutateX(mutateY (mutateZ (C)))$.

\end{itemize}
\end{proof}
 

\subsection{Parallel Repair Algorithm}
\label{sec:algo}

We now describe a parallel genetic repair framework
for vulnerable smart contracts.
The repair framework consists mainly
of eight processes running in parallel ($p_1 || p_2 ||... || p_8$):
the first seven processes ($p_1-p7$) are responsible for generating compilable candidate patches
of the given vulnerable smart contract corresponding to
the search spaces $(S_1- S_7)$
and the last process (process $p_8$)
is responsible for creating concurrent validation processes and selecting the most
preferable patches generated as the base version to be further mutated.
Such parallel repair framework would help
to generate plausible repairs
for vulnerable smart contracts in a much faster way than the repair framework
that generates and validates candidate patches in a traditional sequential order.

\algnewcommand\Sends{\textbf{Sends }}
\algnewcommand\Receives{\textbf{Receives }}
\algnewcommand\Requests{\textbf{Requests }}
\algnewcommand\Break{\textbf{Break}}
\algnewcommand\Terminate{\textbf{Terminate}}

\begin{algorithm} 
\begin{algorithmic}[1]

\While{$p_8$ is running and space is not exhausted}
\State $C_{base} := \Receives{\emph{Base Contract}\ from\ p_8}$
\While{Space $S_1$ is not exhausted}
\State $C_{new} := mutateM(C_{base})$
\If{$C_{new}$ is compilable}
    \State \Sends{$C_{new}\text{ to }p_8$}
    \State \Break
\EndIf
\EndWhile
\EndWhile
\State \Terminate
\end{algorithmic}
\caption{Repair process $p_1$ in our algorithm (process $p_1,\ldots,p_7$ explores part of search space, $S_1, \ldots,S_7$).}  
\label{alg:processP2}
\end{algorithm}

\begin{algorithm}[]
\begin{algorithmic}[1]
\State \textbf{Inputs} :   Vulnerable Contract $C$, Vulnerabilities $U$, Tests $T$
\State \textbf{Inputs} :  Initial Population size \verb+IP+, Generation size \verb+GR+, 
Maximum Population size $P_{\text{size}}$
\State \textbf{Inputs} :  Maximum Gas Usage Bound $L$
\State \textbf{Output} : Set of Plausible Patches
\State \verb+Patches+ := $\emptyset$
\For{$i := 1; i \le IP; i:= i +1$}\Comment{Each iteration executes in parallel}
    \State $C_{new} := \Requests{\text{new mutant of contract\ } C \text{\ from\ } p_1,\ldots,p_7}$
    \State $C_{new}.fitness := \mathit{Eval}(C_{new}, U, T)$
    \State \verb+Patches+ := \verb+Patches+ $\cup \{ C_{new} \}$
\EndFor
\While{(at least one of $p_1,\ldots,p_7$ has not terminated $\land$ timeout not reached)}
    
   \State $plausible := \mathit{Filter\_Plausible\_Patches}({\tt Patches}, U, T, L)$
    \If{plausible != $\emptyset$}
    \State \textbf{return} plausible
    \EndIf
    
    \State Patches := $\mathit{NSGA2Selection}({\tt Patches}, P_{\text{size}})$
    
    \For{$i := 1; i \le GR; i:= i+1$}\Comment{Each iteration executes in parallel}
        \State $C_{current\_best}$ := highest fitness patch from \verb+Patches+
        \State $C_{new} :=  \Requests{\text{new mutant of\ }} C_{current\_best}$ from $p_1,\ldots,p_7$
        \State $C_{new}.fitness := \mathit{Eval}(C_{new}, U, T)$
        \State \verb+Patches+ := \verb+Patches+ $\cup \{ C_{new} \}$
    \EndFor
\EndWhile
\State \textbf{return} $\emptyset$


\end{algorithmic}
\caption{Main Repair Algorithm (process $p_8$ which combines results from processes $p_1,\ldots,p_7$) }  
\label{alg:processP8 }
\end{algorithm}

Each process $p_i\mid i \in \{1, \ldots, 7\}$ is a long-running process. In each iteration, it waits for $p_8$ to send patch generation request. Upon the request is received with a base version of the vulnerable smart contract, the processes will then search for a compilable patch mutated from the received base version. 
The processes use mainly
the set of mutation functions: $mutateM (C), mutateR (C)$ and $mutateI (C)$
which mutate the vulnerable contract
using some mutation operators. Every mutant is then be checked for their syntactic correctness via the use of a compiler. This technique has been shown very effective in early rejecting invalid patches. After the first compilable patch is generated, the process will then send it back to $p_8$ and wait for the next request.
However, since the implementations of
processes $p_1, \ldots, p_7$
are very similar, we present here
the pseudo-code of one of them for brevity
(we choose process $p_1$ that corresponds to the search space $S_1$).
For readability, let us assume that we can get a fresh (new) mutant every time the function $mutateM(C)$ is used.
The pseudo-code of $p_1$ is given in Algorithm \ref{alg:processP2}.
As one can see, Algorithm \ref{alg:processP2} can consider
all possible combinations of random mutations
of the function  $mutateM (C)$ on the contract $C$ until the corresponding patch space $S_1$ is exhausted.

We now discuss the implementation
of the main process $p_8$ (Algorithm \ref{alg:processP8 }).
This process takes as inputs:
the original vulnerable smart contract $C$,
the set of targeted vulnerabilities $U$,
and the set of test cases $T$,
then returns the patches that meets the quality requirements (plausible patches that pass given tests $T$ and do not exhibit given vulnerabilities $U$). At the beginning, we conduct a population bootstrapping that a set of mutants is generated to have the initial set of mutants. The size of the set is controlled by the parameter \verb+IP+ (Initial Population Size). At the time new mutants should be generated, $p_8$ sends requests to the processes $p_1, \ldots, p_7$ (the \textbf{Requests} operation in Algorithm \ref{alg:processP8 }). Whenever one of the processes has generated a new compilable mutant, all other mutant generation processes will stop attempting to generate new mutants and the request is fulfilled. The $\mathit{Eval}$ is used to calculate the fitness value of the patches. The objective functions are defined in \autoref{table:objectives}. Note that all the objective functions are independent from one to the other, the $\mathit{Eval}$ function therefore also issues new concurrent processes to speed up the patch fitness evaluation process. The control flow then enters the main loop. In each iteration, the algorithm first checks if there is already plausible patch existing in the maintained set of patches; this is accomplished by invoking the function {\it Filter\_Plausible\_Patches}. If it exists, this algorithm returns immediately the plausible patch. Otherwise, the maintain set of patches will be trimmed to the size $P_{\text{size}}$ by the NSGA2 population selection algorithm\cite{nsga2} and yet another set of patches will be generated in the similar fashion. The base version used to generate the new set of patches is chosen to be the best patch among all the patches in the maintained set \verb+Patches+. The evaluation of relative quality between patches is based on their fitness values. In each iteration of the main loop, the number of new patches will be generated is determined by the parameter \verb+GR+ (Generation Rate).

We employ a timer in $p_8$ (not shown in pseudo-code for simplicity)
which will be used to enforce termination of the process in
case the time spent in the search process exceeds the bound $MaxBound$.
The bound $MaxBound$ should be chosen while taking into
consideration the number of test cases, the size of the buggy program,
and the estimated number of mutants in the search space assigned to the process.
Note that processes work independently and terminate whenever a plausible patch is found or that the timer is fired.

\paragraph*{Objectives or Fitness Functions} 
As mentioned earlier, the size 
of the search space can be extremely large
even for programs whose source code size is small.
Recall that the search space grows exponentially
with the considered lines of code
and hence the efficiency and performance
of the genetic repair algorithm needs to be improved
when examining candidate patches in the generated search space.
While the parallel repair algorithm
splits the large search space
into smaller sub-spaces which improves considerably the patch generation process,
the search sub-spaces can be still huge
to be exhaustively explored in a reasonable time budget.
The goal of the employed fitness functions
is to guide the search towards plausible repair.
We therefore integrate four fitness functions
(objectives) with the patch generation process.
The objectives are classified into primary objectives and secondary objectives. Primary objectives are related to the functional or correctness properties of the patch,
while secondary objectives are related to the non-functional properties of the patch. 
The two main functional correctness objectives are the number of targeted vulnerabilities and the number of failing test cases. The number of targeted vulnerabilities can be retrieved from any smart contract vulnerability detector (e.g. Oyente \cite{Luu2016}) while test cases can be provided by the vulnerable contract developers.
The secondary properties
or non-functional properties include the number of mutation operators applied on the generated patch and the gas usage or the cost of the patch. 
 The designated fitness functions measure how many of desired functional and non-functional requirements a generated mutant meets.
The mutation distance of the generated mutant
from the original vulnerable contract is measured by counting the number of times the mutation operators applied to the generated mutant.
This can be used to measure the simplicity of the generated mutant.
The average gas usage is compared by the methodology described in \autoref{GasAnalysis}.
The two secondary objectives
are considered only when the generated patch is valid
(fixed all targeted vulnerabilities and passes all test cases).
Note that we give higher preference to a patch that
fixes all detected vulnerabilities and passes all test cases with lower
average gas usage and smaller number of syntactical changes 
 w.r.t. the original vulnerable contract.
 We summarize these objectives (fitness functions) in Table \ref{table:objectives}.

\begin{table}[H]
\caption{Objectives (fitness functions) used when generating patches}
\label{table:objectives}
\scalebox{0.9}{
\begin{tabular}{|l|l|l|l|}
\hline
    \textbf{Description of objective} & \textbf{Objective Purpose} & \textbf{Objective Type} & \textbf{Importance} \\ \hline
     Number of targeted vulnerabilities & Patch correctness  &  Functional & Primary \\ \hline
     Number of failing test cases &  Patch correctness &  Functional&  Primary \\ \hline
     Gas consumption  &  Patch gas optimization &  Non-functional&  Secondary \\ \hline
    Mutation operation distance & Patch simplicity & Non-functional&  Secondary \\ \hline
\end{tabular}
}
\end{table}


\bigskip
\section{CHOOSING PATCH WITH LOWER GAS CONSUMPTION} \label{GasAnalysis}

One of the key challenges we encounter in this work
is how to compare efficiently the average gas usage between the original contract
and the repaired contract and how to compare the average gas usage
of different generated patches of a given vulnerable contract.
In general, the gas cost of a smart contract depends on 
 a number of parameters
including memory cost, stack cost, and storage cost
in addition to the instructions' costs.
Hence, the gas consumption of a given path $\pi$
in a smart contract $SC$ can be a non-constant.
It should be therefore described 
as a parametric  formula
that takes into consideration
the parameters that affect the 
gas consumption of the path. We call the described parametric formula as \emph{gas formula}.
 
 To compare the average gas usage of two smart contracts, we propose the notion of \textit{gas dominance}.
The goal of the introduced gas dominance notion
is to rank edited contracts (generated repairs of vulnerable contracts)
based on their corresponding gas formula as an estimation on the relative average gas usage. This estimation is required as we cannot predict in advance the true average gas usage over their lifespan. Such a ranking approach can be used to select a low-cost repair for a vulnerable smart contract from the set of proposed repairs generated by the parallel repair algorithm.


\subsection{The Gas Dominance Relationship}

When formalizing the gas usage of smart contracts, 
we choose  the specification of the gas cost function in the current Ethereum virtual machine specification  (version EIP-150) \cite{wood2014ethereum} at the time of writing of this paper. From a high-level perspective, the gas usage of a single invocation to the smart contract depends on the user input to the  smart contract, the blockchain environment, and the code of the smart contract. The gas usage of an execution  (a transaction) to a  smart contract  is the sum of the gas usage of each executed instruction along the execution path.
Formally, the gas cost function $C$ of an instruction $inst$ can be defined as

\begin{equation}
\label{equ:GU_inst}
    C(\sigma_{inst}, \mu_{inst}, I) = GU_{\text{OPCODE}_{inst}}(\sigma_{inst}, \mu_{inst}, I)
    + GU_{\text{New Memory}}(\sigma_{inst}, \mu_{inst}, I)
\end{equation}

where $\mathbf{\sigma_{inst}}$ is the blockchain world state  before the instruction $inst$ is executed and $\mathbf{\mu_{inst}}$ is the machine state  before $inst$ is executed,  the operation code $\text{OPCODE}_{inst} = I.code[\mu_{pc}]$ is a property of the execution environment $I$ indexed by a program counter $\mu_{pc}$, and $GU_{\text{OPCODE}_{inst}}$ is the gas formula associated to the operation code of $inst$ and $GU_{\text{New Memory}}$ is the gas usage formula associated to the expansion of machine memory when executing the instruction $inst$. For more technical details about the definition of the gas cost function, we refer the reader to \cite{wood2014ethereum}.
 
The total gas usage of an invocation (in the form of a single transaction) with the execution information specified in $I$ can be defined as a gas function corresponding to the visited contract path triggered by the inputs:

\begin{equation}
\label{equ:GU_path}
GU_{\text{path}}(\sigma_p, \mu_p, I) = \sum_{inst \in \mathbf{Insts}} C(\sigma_{inst}, \mu_{inst}, I)
\end{equation}

where $\mathbf{Insts} = (inst_0, inst_1, inst_2, \ldots)$
the sequence of instructions in the execution path determined by $\sigma_{p}$, $\mu_{p}$ and $I$, and $\sigma_{p} = \sigma_{inst_0}$, and $\mu_{p} = \mu_{inst_0}$. 
For a smart contract with $k$ execution paths, we construct $k$ gas usage
functions, e.g. $GU_{path_1}, \ldots,  GU_{path_k}$ . 
We can then express the total gas usage of a smart contract $SC$ over its lifespan as follows:

\begin{equation}
GU_{\text{lifespan, SC}} = \sum_{t \in trans} GU_{trans}{(t)}{(\sigma_{t}, \mu_{t}, I_t)}
\end{equation}

where $trans$ is the set of transactions to smart contract (denoted by $SC$) over its lifespan (the history of transactions of $SC$), and $\sigma_t$, $\mu_t$ and $I_t$ are the world state, machine state and execution environment respectively when the first instruction of the invocation corresponding to transaction $t$ was executed. We introduce a new higher order function $GU_{trans}$ here that maps a transaction to its corresponding gas usage function. Suppose the execution path is $\pi$ for the transaction $t$, then $GU_{trans}(t) = GU_{path_{\pi}}$.

Given two repaired versions $SC_a$ and $SC_b$ for a vulnerable
smart contract $SC$ addressing the same vulnerabilities, we then favor the version with lower lifespan gas usage.
However, since the future blockchain world state and the user inputs to $SC$ can be of any possible combination which are generally unknown in advance, concrete lifespan gas usage of patched versions cannot be used to compare effectively the average gas usage of patches. We therefore propose to use what we call \emph{gas dominance} as a method to compare the relative gas-efficiency between two patches by comparing the expected gas usage functions of them.
So that for a  given a smart contract $SC_a$ with $k$ execution paths, we can express the expected gas usage of $SC_a$ as follows:

\begin{equation}
\label{equ:GU_expected}
GU_E(SC_a)(\sigma, \mu, I) = \sum_{i = 0}^{k}  \mathbf{P}_i * GU_{path_i}(\sigma, \mu, I)
\end{equation}

where $\mathbf{P}_i$ is the probability of $path_i$ being visited by an arbitrary execution of $SC_a$, $GU_{path_i}$ is the gas usage function corresponds to program path $path_i$.
For the cases where the contract paths invoke external functions, we need  to include the gas usage introduced by the external function invocations in the equation of $GU_E(SC_a)$ of the contract. 

\begin{definition} (\textbf{Gas Dominance Relation}). Given two smart contracts $SC_a$ and $SC_b$, we say $SC_a$ gas dominates $SC_b$ (denoted by $SC_a >_{g} SC_b$) if and only if $GU_E(SC_a) \le GU_E(SC_b)$ for all inputs and $GU_E(SC_a) < GU_E(SC_b)$ for at least one input to the smart contracts. 

Formally, 
\begin{equation}
\begin{aligned}
    SC_a >_{g} SC_b
        \iff \forall \sigma,\mu,I ~{( GU_{E_{a}}(\sigma, \mu, I) \le GU_{E_{b}}(\sigma, \mu, I) )} \land \\
        \exists \sigma,\mu,I ~{( GU_{E_{a}}(\sigma,\mu,I) < GU_{E_{b}}(\sigma,\mu,I) )}
\end{aligned}
\end{equation}

where $GU_{E_{a}} = GU_E(SC_a)$ and $GU_{E_{b}} = GU_E(SC_b)$
\end{definition}

The gas dominance relation has the following properties:

\begin{property}[Irreflexive]
    For all smart contracts $SC$, they do not gas dominate themselves. That is, $SC$ must not gas dominate $SC$.
\end{property}

\begin{property}[Asymmetric]
For two arbitrary smart contracts $SC_a$ and $SC_b$, if $SC_a$ gas dominates $SC_b$, then $SC_b$ must not gas dominate $SC_a$.
\end{property}

\begin{property}[Transitive]
For three arbitrary smart contracts $SC_a$, $SC_b$ and $SC_c$, $SC_a$ gas dominates $SC_b$ and $SC_b$ gas dominates $SC_c$, then $SC_a$ must gas dominate $SC_c$.
\end{property}

\subsection{Lightweight Approximation for Determining Gas Dominance Relationship}

In general, determining the gas dominance relationship between two smart contracts can be a computationally complex task and practically infeasible because the possible input space is generally too enormous.
We therefore develop a lightweight approximation approach  based on the notion of function dominance.  We say that one gas formula dominates another formula
if the magnitude of the ratio of the first formula
to the second increases without bound as the inputs increase without bound.
There are different  ways to compare the gas consumption  
of two smart contracts and we describe here two approaches.

Given two contracts $SC_a$ and $SC_b$, we first construct the expected gas usage
formulas for $SC_a$ and $SC_b$, namely  $GU_{E} (SC_a)$ and  $GU_{E} (SC_b)$.
We then transform the equations  $GU_{E} (SC_a)$ and  $GU_{E} (SC_b)$
into polynomial expressions. Due to the fact that there might be terms containing non-polynomial functions,  we use  a {\em substitution mapping} to transform
the gas formula into a polynomial expression. The substitution mapping is constructed as follows.
\begin{enumerate}
    \item For all monomial terms, they are unchanged.
    \item For other terms, the coefficient remains unchanged while the other parts of the term is mapped to a unique fresh variable.
\end{enumerate}
All common non-monomial terms in $GU_{E} (SC_a)$ and  $GU_{E} (SC_b)$ are mapped to the same fresh variable, that is, variable binding of the fresh variables are maintained for the substitution mappings e.g. if the formula $x^2 + sin(x)$ is substitution mapped to $x^2 + y$, the formula $x^2 + cos(x) + 3sin(x)$ should be substitution mapped to $x^2 + z + 3y$.
A polynomial can be expressed as a sum of monomials where each monomial is called a term. The degree of the polynomial is the greatest degree of its terms. 
We denote the resulting polynomial equation for  $GU_{E} (SC_a)$
by  $GU_{E}^{poly} (SC_a)$ and the resulting polynomial equation for  $GU_{E} (SC_b)$ by  $GU_{E}^{poly} (SC_b)$.
We then rearrange and  simplify the resulting polynomial equations $GU_{E}^{poly} (SC_a)$  and $GU_{E}^{poly} (SC_b)$ as a sum of monomials.
Let $M_{SC_a}$ and $M_{SC_b}$ be the sets of monomials in $GU_{E}^{poly} (SC_a)$  and $GU_{E}^{poly} (SC_b)$.
We can determine the gas dominance relationship between $SC_a$ and $SC_b$ as follows (apply in order).
\begin{enumerate}
    \item If $|M_{SC_a}| \ne |M_{SC_b}|$, then $SC_a$ and $SC_b$ are not gas dominating each other.
    \item Let $V_{SC_a}$ and $V_{SC_b}$ be the vectors of coefficients of  $GU_{E}^{poly} (SC_a)$ and  $GU_{E}^{poly} (SC_b)$ respectively so that the order of elements of $V_{SC_a}$ and $V_{SC_b}$ should be aligned according to the same corresponding monomials.
    \begin{enumerate}
        \item If $V_{SC_a} \le V_{SC_b}$ (all elements in $V_{SC_a}$ are less than or equal to the corresponding elements in $V_{SC_b}$ and $V_{SC_a} \ne V_{SC_b}$), then $SC_a >_{g} SC_b$.
        \item If $V_{SC_a} = V_{SC_b}$, then $SC_a$ and $SC_b$ are not gas dominating each other.
        \item If $V_{SC_a} \ge V_{SC_b}$ (all elements in $V_{SC_a}$ are greater than or equal to the corresponding elements in $V_{SC_b}$ and $V_{SC_a} \ne V_{SC_b}$), then $SC_b >_{g} SC_a$.
        \item if none of the above conditions hold, then $SC_a$ and $SC_b$ are not dominating each other.
    \end{enumerate}
\end{enumerate}


\subsection{Integrating Gas Dominance Relationship into Genetic Patch Search Process}

The above defined gas dominance relationship is for comparing the relative average gas consumption between two versions of the vulnerable contract. To enable the comparison among multiple patched versions of the original vulnerable contract, we here define the notion of \emph{gas dominance level}, as defined in the following.
\begin{definition} (\textbf{Gas Dominance Level}). Given a set of smart contracts, non-dominated sorting \cite{nsga2} is performed based on the gas dominance relationship. The gas dominance level of an arbitrary smart contract in the set is defined as its ranking in the non-dominated sorting result.
\end{definition}
The multi-objective genetic algorithm can now use the gas dominance level as one of the objectives, which serves to implicitly capture the effect of patches on the gas consumption (without having to compute the gas consumption directly).

\subsection{Accelerating Gas Comparison by Generating 
Reduced Gas formulas} \label{sec:accelartion}

As described in the preceding, to compare the gas usage of two contracts we
need first to synthesize gas formulas for the set of feasible paths in  each contract.
Note that the number of gas formulas generated for each patch
can affect the computational complexity of the gas comparative approach dramatically.
Suppose that the parallel genetic algorithm
generates three plausible patches for a vulnerable contract $C$,
namely $C_1, C_2$ and $C_3$. 
However, to compare efficiently the gas usage of the contracts  $C_1, C_2$ and $C_3$
we only need to synthesize gas formulas for
the set of different paths in the three contracts.
It is sufficient to  conduct a comparison between reduced versions of these contracts by skipping joint or common paths.
This helps to reduce the computational complexity
of the comparative approach.







\begin{remark}
Syntactically identical paths among contracts share the same gas formula
and therefore can be safely skipped during comparison. 
\end{remark}


\begin{definition} (\textbf{Classifying paths in  contracts}).
Let $C$ be a vulnerable smart contract
and $C^{'}$ be a repaired versions of $C$ 
obtained by the parallel repair algorithm.
A feasible path $\pi$ in $C^{'}$  can
be classified into one of the following categories 

\begin{itemize}

\item $\pi$  is a repaired path of some paths in $C$, or

\item $\pi$  is a new path w.r.t. the set of feasible paths in $C$, or

\item $\pi$  is a joint or common path between $C$ and $C^{'}$.

\end{itemize}

\end{definition}

Note that a patch introduces to a given vulnerable smart contract may
trigger a new set of paths that were infeasible in the original vulnerable smart contract.
Thus, a repaired version of a contract may have new set of behaviors w.r.t. the original contract.
This may happen for example when the patch updates an expression
in a conditional statement in the original vulnerable contract.
The advantages of distinction between the above three classes of paths
are two-fold.
First, it helps to reduce the number of paths that need to be considered
when comparing the contracts and hence the number of
gas formulas that need to be synthesized. Second, it helps to reduce
the complexity of the final gas formulas of the contacts
being compared.
Note that since we use a genetic algorithm
based on three mutation operators (\gls{move}, \gls{insert}, and \gls{replace}), we can easily then classify paths in
the contracts being compared into three categories: repaired paths,
joint paths, or new paths.
Typically, we can identify the locations of buggy statements in the contract 
and we can augment the repairing algorithm to label the locations
of statements that have been influenced by the deployed patch.
This facilitates the classification of paths
in the generated repaired contract w.r.t. the original contract.

We now turn to describe  an acceleration technique 
that can be applied before conducting the actual
comparison between two similar contracts $C$ and $C^{'}$.
Let us denote the set of feasible paths in the two
contracts by $ \Pi_{C}$ and $ \Pi_{C^{'}}$.
The goal of the acceleration technique is  to generate reduced versions 
of the contracts $C$ and $C^{'}$ as follows:

\begin{enumerate}

\item Compute the sets of paths that are unique in each contract as follows

$$
\mathit{Diff}(C, C^{'}) = (\Pi_{C} \setminus \Pi_{C^{'}})
$$
$$
\mathit{Diff}(C^{'}, C) = (\Pi_{C^{'}} \setminus \Pi_{C})
$$

\item  Synthesize a gas formula for each path in the sets $\mathit{Diff}(C, C^{'})$ and $\mathit{Diff}(C^{'}, C)$ using Equation (\ref{equ:GU_path}) and then compute the final gas  formula by summing the resulting gas formula using Equation (\ref{equ:GU_expected}).


\item Compare the resulting gas formulas 
using the comparative approach described at Section \ref{GasAnalysis}.

\end{enumerate}

Comparing the gas usage
of two contracts using their reduced versions 
(i.e., versions obtained by skipping joint paths or repaired
paths whose gas formulas are equivalent)
preserves soundness, as described in the following theorem.

\begin{theorem} (\textbf{Soundness of reduction}).
Let $C$ be a vulnerable smart contract and $C^{'}$
be a repaired version of $C$. 
Let also $G(C)$ and $G(C^{'})$ be  gas formulas
for $C$ and  $C^{'}$ respectively
and $G(C_R)$ and $G(C_R^{'})$
be gas formulas for reduced versions of $C$ and  $C^{'}$
obtained as described at Section \ref{sec:accelartion}.
$G(C_R)$ dominates $G(C_R^{'})$ if and only if
$G(C)$ dominates $G(C^{'})$.
\end{theorem}

\begin{remark}
 (\textbf{Effectiveness of reduction}).
The  accelerated comparative approach of smart contracts
has lower computational complexity than the
non-accelerated comparative approach.
The amount of reduction on the computational complexity that can be obtained depends on the number of joint and repaired paths
in the contracts being compared that can be skipped safely during the comparison
(i.e., without adversely affecting the outcome of comparison).
\end{remark}

The number of generated gas sub-formulas (for paths)
and the complexity of the final gas formula (for the contract)
can be significantly reduced if the acceleration approach is employed.
This is crucial as synthesizing gas formulas for paths
can be an expensive step specially for paths with cyclic behavior.
Note that comparing reduced versions of contracts using simplified or reduced gas formulas that consider only different paths
in the two contracts does not affect the soundness of the analysis.
This is mainly due to the observation that only the set of different
paths in the contracts can make the gas consumption of
 a contract dominates the other.

\section{Implementation} \label{sec: implementation}

In this section, we describe the implementation of the \textsc{ SCRepair} tool, as well as the setup of the experimental evaluation (the results from the experiments will appear in the next section).

\subsection{Prototype implementation}

To evaluate our presented repair approach
for vulnerable smart contracts, we have implemented
a tool called \textsc{ SCRepair}.
The tool interacts and takes in inputs from the  smart contract security analyzers 
Slither \cite{slither} and Oyente \cite{Luu2016} in order to analyze and detect 
security vulnerabilities (if any) in the subject smart contracts. 
The tool Slither is a static analysis based detector which is able to reliably detect various vulnerabilities within a short time due to the lightweight nature of static program analysis. While being lightweight, it has been showing promising accuracy practically and used by the industry. On the other hand, the other supported vulnerability detector tool Oyente is a symbolic execution tool that works directly with Ethereum virtual machine code. It is able to detect many commonly occurring security flaws of Ethereum smart contracts. Notably, it can detect integer overflow and transaction order dependence vulnerabilities for which are difficult to be detected with pure static program analysis due to the need to reason about dynamic program behaviors. The fault localization information provided by the both vulnerability detection tools is used for fix localization in our search-based repair engine.

Since our repairing approach 
aims not only to fix the vulnerability but also
to optimize the gas usage of the patched 
vulnerable smart contract, we have built a gas analyzer based on Oyente in a way such that it can generate the information for determining the approximate gas dominance relationship. For determining gas dominance, it has a gas usage model extended from the Oyente implementation which is closer to the actual Ethereum virtual machine's gas model.

In Fig. \ref{fig:System-Component} we give  the
schematic diagram of our smart contract repair tool in which we describe the main components of the tool.
The tool consists  of five units: the vulnerability
detector, the test case executor, the gas ranker, the patch generator, and the main process controller. Units except the main process controller are executed in worker processes/threads. The main process controller unit manages all the worker processes and threads to perform the repair task in a concurrent manner. 

We have also implemented a variant of SCRepair with an unguided random search repair algorithm called \textsc{SCRepair-URS}. This implementation mostly reuses the SCRepair implementation except that the genetic search mechanism is removed. The patch evaluation can now be terminated early as long as it has sufficient information to assert that the patch under validation is not plausible ({\em e.g.}, as soon as a vulnerability is detected or a test case failed under the patched version). This acts as an optimization which is not possible to apply to genetic repair algorithm since early termination does not give the algorithm an evaluation of the fitness functions ({\em e.g.}, does not generate total number of test cases failed).

\begin{figure}
    \centering
    \includegraphics[width=\textwidth]{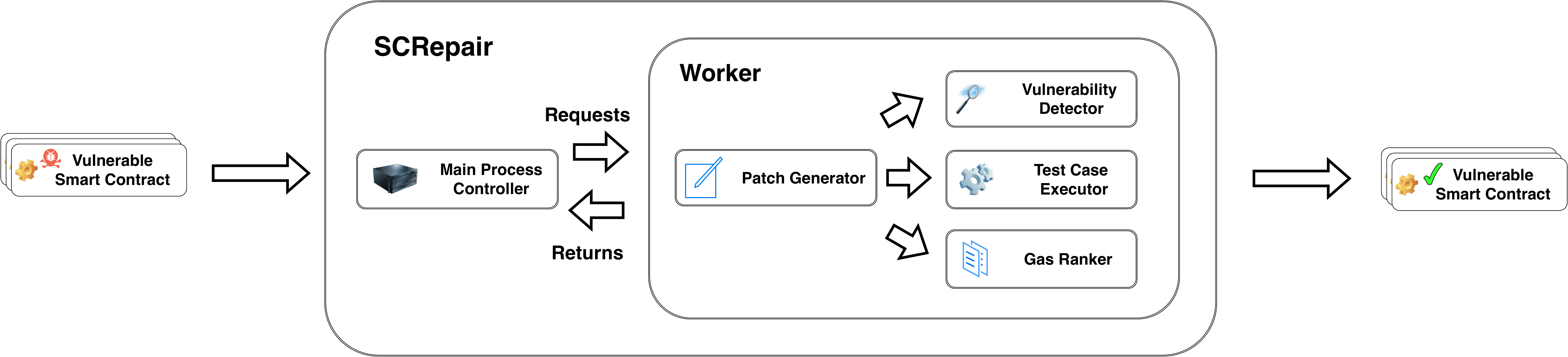}
    \caption{The schematic diagram of the SCRepair tool}
    \label{fig:System-Component}
\end{figure}


\subsection{Etherscan Vulnerable Dataset (EV-DS)} 
\label{sec:dataset}

To evaluate our repair approach, we have constructed 
a dataset of vulnerable smart contracts taken mainly from Etherscan as a proxy to real-world deployed \gls{sc} source code.
Etherscan is a well-known block explorer, search, API and analytic platform for
Ethereum mainnet, which is the main network wherein actual transactions
of smart contracts take place on a distributed ledger. 
A large amount of information related 
to the smart contracts can be extracted from
Etherscan, e.g., deployment address, verified source code, byte-code and application
binary interface (ABI) of deployed contracts.
From Etherscan we first collected all 38,225 available smart contract source code files which correspond to smart contracts deployed before 1st August 2018. After de-duplication,
34,400 files remained.
These source files are then analysed using the tool Oyente.
We obtained 2,752 vulnerable smart
contracts with different types of vulnerabilities. Four types of vulnerabilities have been detected on this dataset: \gls{tod},
\gls{re}, \gls{ed}, and \gls{io}.

The \gls{tod} happens when the user of a smart contract assumes a particular state of a contract, which may not exist when his transaction is processed potentially leading to malicious behavior. Reentrancy vulnerability is probably the most widely known vulnerability as it led to the infamous DAO attack (causing a loss of 60 million US dollars in June 2016). 
\Gls{re} happens when a contract is called by another contract
so that the original contract has to wait for the call to finish. 
This intermediate state can be exploited.
The contracts may suffer also from the so-called exception disorder (\gls{ed}) vulnerability
where the contract does not check explicitly whether 
the send operations have been completed successfully. Integer overflow (IO) is a common problem across all systems which could be used to modify the program state in an unwanted manner by deliberately providing large numbers as inputs leading to incorrect results being calculated in mathematical operations.

While the selection of smart contracts shown in 
Tables \ref{table: regressionTestCases} and \ref{table:RQ1-result} has been made randomly from the  dataset EV-DS, 
we have considered some key criteria
when selecting these smart contracts. 
The two main criteria we considered
are: (1) the size and complexity of the vulnerable smart contract
measured mainly in terms of the number of lines in the contract,
and (2) the popularity and the number of available
transactions of each vulnerable smart contract. Therefore, we filter the dataset according to the following rules before random sampling:
\begin{enumerate}
    \item Number of lines of code (exclude comments) > 30
    \item Number of transactions (at the time data collection was performed) > 30
\end{enumerate}

\subsection{Test Case Generation for EV-DS Dataset}


Since test cases for smart contracts are in general not available on the blockchain and  the authors of the deployed smart contracts are also not contactable \cite{Luu2016}, we therefore use a novel method to generate regression test cases from the available transactions to the subjects smart contracts on the blockchain. For every transaction (denoted by $t$) to the subject smart contract, we capture the inputs and the changes to the blockchain state during the execution of $t$ which then considered as the inputs and expected behaviors of the generated regression test case. A generated regression test case for a transaction $t$ contains the following elements:

\begin{enumerate}
    \item \label{tc:prevState} Blockchain state before executing the transaction $t$.
    \item \label{tc:func} The function being invoked and the corresponding argument values.
    \item \label{tc:aftState} Blockchain state after executing the transaction $t$.
    \item \label{tc:ret} The return values of invoked functions.
\end{enumerate}

However, as the whole blockchain state can be very huge (in the magnitude of terabytes), it is impractical to simply store relevant versions of the blockchain state. To address this issue, we only capture relevant states of  the Ethereum  accounts in the blockchain  before and after the execution of the transaction $t$. The generation of each regression test case is then run against the original vulnerable smart contract  to check the validity of the newly generated test case. During the test case generation process, we have set a timeout bound of 5 minutes for the execution time of each regression test case. Regression tests requiring longer time are terminated and discarded. \autoref{table: regressionTestCases} shows the number of regression test cases generated for each subject contract. The generated regression test cases are then used in the automated repair experiments.

\begin{table}[]
\caption{EV-DS dataset subject smart contract information}
\scalebox{0.8}{
\begin{tabular}{|l|c|c|c|c|}
\hline
\textbf{Name of Contract} & \textbf{\# Lines} & \textbf{\# Transactions} & \textbf{\# Regression} & \textbf{Supported by} \\
 & & & \textbf{Tests} & \textbf{SCRepair}
 \\ \hline
Autonio ICO & 330 & 34 & 31 & Yes \\ \hline
Airdrop & 62 & 147 & 7 & Yes \\ \hline
Banana Coin & 117 & 360 & 24 & Yes \\ \hline
XGold Coin & 272 & 308 & 304 & Yes \\ \hline
Flight Delay Issuance & 429 & 80 & 1 & No \\ \hline
Hodbo Crowdsale & 268 & 36 & 18 & Yes \\ \hline
Lescoin Presale & 351 &115 & 107 & Yes \\ \hline
Classy Coin & 217 & 574 & 495 & Yes \\ \hline
Yobcoin Crowdsale & 481 &515 & 435 & Yes \\ \hline
Classy Coin Airdrop & 49 & 137 & 4 & Yes \\ \hline
OKO Token ICO & 232 & 179 & 173 & Yes \\ \hline
ApplauseCash Crowdsale & 407 & 43 & 42 & Yes \\ \hline
HDL Presale & 239 &94 & 93 & Yes \\ \hline
Privatix Presale & 179 & 78 & 11 & Yes \\ \hline
MXToken Crowdsale & 186 & 56 & 37 & Yes \\ \hline
EthereumFox & 77 & 493 & 491 & No \\ \hline
dgame & 42 & 302 & 108 & Yes \\ \hline
Easy Mine ICO & 351 & 1339 & 491 & Yes \\ \hline
Siring Clock Auction & 978 & 1641 & 2 & Yes \\ \hline
Government & 83 & 502 & 366 & No \\ \hline
\end{tabular}
}
\label{table: regressionTestCases}
\end{table}

\subsection{Factors Affecting our Repair Algorithm}

Before discussing the research questions that we developed
to evaluate the presented
genetic repair algorithm, we first summarize the key factors
that affect the correctness and efficiency of our genetic repair algorithm.

\begin{enumerate}
    
\item \textbf{Quality of test suite and vulnerabilities}.
The quality of provided test suite 
for a given vulnerable smart contract has a major impact
on the genetic repair algorithm. 
Recall that a mutant is considered as a repair
when all available test cases pass and the vulnerability detector does not report any vulnerability found. 
In our experiments, we constructed the test suite with a script to convert past block-chain transactions as positive test-cases as described in \autoref{sec:dataset}. 
The vulnerabilities detected by a smart contract checker like Oyente and Slither constitute the negative behavior that the generated patches should avoid.

\item \textbf{Timeout allocated to the algorithm}.
A feasible exploration of the search space (candidate patches) depends heavily on the amount of resources allocated to the genetic algorithm.  In general, the size of the generated search space of a given vulnerable contract depends on multiple factors including: (i) the size and complexity of the contract being repaired, (ii) the number of buggy statements in the contract, and  (iii) the mutation operators used by the algorithm.
However, the number of mutants that can be examined during the search is limited to the time budget allocated to the  algorithm. 
The bigger the time budget,
the higher the probability to produce a plausible patch.

\item \textbf{The consideration of gas consumption of patches}.
Considering the gas when searching for plausible patches
of a vulnerable smart contract can be of great benefit.
First, it can help to generate a low-cost repair for a given
vulnerable smart contract by comparing the gas consumption of generated
patches and selecting the one with low average cost.
Second, it can be used to optimize the efficiency of the genetic search algorithm in various ways. For example, it can be used to detect and discard  infeasible patches early. Note that a patch can be a plausible patch (passes the test-cases) but infeasible to be deployed on a real blockchain. This happens when the generated patch consumes a significantly large amount of gas and thus leads to expensive transactions. 
To  reduce the computational complexity
of the algorithm, one might need to maintain during the genetic search
the best known low-cost average gas usage (let us call it $g_{max}$) of a plausible patch. Then when a new  plausible patch is found that has lower average cost,
the bound $g_{max}$ will be updated accordingly. 
The bound $g_{max}$ can be updated on-the-fly during the search and used  to discard infeasible patches early without necessarily examining the entire test suite.

\item \textbf{The number of genetic mutation operators used by the algorithm}.
Note that the size of the search space that needs to be examined
when searching for a plausible patch for a vulnerable contract
can be extremely large.
Recall that the search space of a given vulnerable smart contract
is generated by mutating (buggy) statements in the contract.
Hence, the size of the generated search space grows
exponentially w.r.t. the number of considered lines
in the contract and the number of mutation operators.
The smaller the number of the mutation operators,
the smaller the size of the search space and the faster the algorithm.
However, reducing the number of the mutation operators may reduce significantly the capability of the algorithm to produce plausible patches.

\item \textbf{The state space search order}.
As the search space grows, the organization of mutants or candidate patches into sub-spaces becomes more critical to the efficiency of the  algorithm.
In general, there is no specific search
 strategy that one can follow
when examining the candidate patches
of a given vulnerable smart contract.
The search can be purely sequential and random
or it can be parallelized based on the semantics
of the mutation operators. 
However, as expected,
the search can be optimized by taking into consideration
some interesting factors including the semantics of the bug,
the semantics of the mutation operators,
and the gas consumption of generated patches.
\end{enumerate}

As one can see from the aforementioned factors,
 the correctness and efficiency of the genetic algorithm can be evaluated under many different settings.
For example, one might wonder how does the algorithm perform
when enabling/disabling the gas calculation
of generated patches or when increasing/decreasing
the size of test suite or the amount of time budget allocated to the genetic algorithm.
In this work, we choose to evaluate the correctness and efficiency of the genetic algorithm by considering five key {\em research questions}. The goal of the research questions is to evaluate the presented parallel genetic repair algorithm and to understand and draw conclusions about the factors affecting the correctness and quality of generated patches.

\section{Research Questions and Experimental Results}
\label{sec:RQ}

\paragraph*{Responsible Disclosure}
 We decide to publish the dataset for open science. The blockchain system is decentralized so even if we want to contact the owner we cannot find them. 
 
\paragraph*{Experimental Setup} 
 We run our tool on a single Amazon AWS EC2 instance c5.24xlarge which has 192GB of RAM and AWS-customized 2nd generation of Intel Xeon Scalable processor with 96 CPU execution threads allocated. Our repair algorithm and its implementation can be run on a compute cluster of multiple computing nodes. However, we run our experiments on a single node in this work for the simplicity and sake of financial budget.
Among the 20 vulnerable \gls{sc} subjects, our implementation prototype was able to handle 17 of them. The remaining 3  have syntax constructs that are currently unsupported or the version of Solidity used in the implementation of these contracts is too old to be supported. Therefore, we carried out our experiments on the 17 supported subjects. For increasing the variety of vulnerabilities being considered while avoiding the expensive cost of symbolic execution, we have employed Slither as the vulnerability detector for the first fourteen subjects and Oyente for the remaining subjects. We limit the scope of targeted vulnerabilities in our experiments to have more focused study to the following vulnerabilities: \gls{ed}, \gls{re}, \gls{io}, \gls{tod}. Our implemented gas analyzer is employed for determining the gas dominance relation between patches. In our experiments, the maximum gas usage bound $L$ is not specified since a reasonable value is subject to the concrete usage of the subject smart contracts from the viewpoint of the original developers.


\begin{table}
\caption{Results for   RQ\ref{RQ-main}, showing efficacy of patching. \\(Timeout is 1 hour, OOM denotes out of memory)}
\scalebox{0.7}{
\begin{tabular}{|l|c|c|c|c|}
\hline
\begin{tabular}{c}Name of Contract\end{tabular} & \begin{tabular}{c}
Vulnerabilities\\Discovered\end{tabular} & \begin{tabular}{c}SCRepair\\ Vulnerabilities Repaired\\ (Correct/Plausible)\end{tabular} &
\begin{tabular}{c}SCRepair-URS\\ Vulnerabilities Repaired\\ (Plausible)\end{tabular} & 
\begin{tabular}{c}Avg Run Time (mins)\\ SCRepair/\\SCRepair-URS\end{tabular} \\ \hline
Autonio ICO &  \gls{ed}(1) & \gls{ed}(0/1) &\gls{ed}(1) & 3/0.4 \\ \hline
Airdrop &  \gls{ed}(4) & \gls{ed}(3/4) & None & 8/OOM \\ \hline
Banana Coin &  \gls{ed}(1), \gls{re}(1) & \gls{ed}(1/1), \gls{re}(1/1) &\gls{ed}(1), \gls{re}(1)& 16/15.8 \\ \hline
XGold Coin &  \gls{ed}(2) & \gls{ed}(2/2) &None& 12/60 \\ \hline
Hodbo Crowdsale  &\gls{ed}(2) & \gls{ed}(2/2) & \gls{ed}(2)& 22/19.8 \\ \hline
Lescoin Presale  & \gls{ed}(2) & \gls{ed}(1/1) &\gls{ed}(2)& 2/14.8 \\ \hline
Classy Coin &   \gls{ed}(1), \gls{re}(1) & None & None & 29/60\\ \hline
Yobcoin Crowdsale &  \gls{ed}(2), \gls{re}(1) & \gls{ed}(1/1), \gls{re}(1/1) & None & 60/OOM\\ \hline
Classy Coin Airdrop &  \gls{ed}(2) & \gls{ed}(1/2) &\gls{ed}(2)& 1/3 \\ \hline
OKO Token ICO &  \gls{ed}(4), \gls{re}(2) & \gls{ed}(1/1), \gls{re}(1/2) &None& 60/OOM \\ \hline
ApplauseCash Crowdsale &  \gls{ed}(2), \gls{re}(1) & \gls{ed}(1/1) &None& 60/OOM \\ \hline
HDL Presale &  \gls{ed}(3) & \gls{ed}(3/3) &None& 55/51.8 \\ \hline
Privatix Presale &   \gls{ed}(1) & \gls{ed}(1/1) &\gls{ed}(1)& 2/2.8 \\ \hline
MXToken Crowdsale & \gls{ed}(1) & \gls{ed}(1/1) &\gls{ed}(1)& 30/5.2 \\ \hline
dgame &  \gls{io}(3), \gls{tod}(1) & None &None& 2/OOM \\ \hline
Easy Mine ICO &  \gls{io}(6), \gls{tod}(1) & None &None& 60/50.4\\ \hline
Siring Clock Auction & \gls{io}(3) & \gls{io}(0/1) & None & 4/60 \\ \hline
Total&\begin{tabular}{c}\gls{ed}(28),\gls{io}(12),\gls{re}(6)\\\gls{tod}(2), sum: 48\end{tabular}&\begin{tabular}{c}\gls{ed}(18/21),\gls{re}(3/4)\\\gls{io}(0/1), sum: 21/26\end{tabular}&\begin{tabular}{c}\gls{ed}(10),\gls{re}(1)\\sum: 11\end{tabular}& \\ \hline
\end{tabular} }
\label{table:RQ1-result}
\end{table}

\refstepcounter{RQ}\label{RQ-main}
\subsection*{RQ\theRQ: How effective is the repair algorithm at fixing detected bugs?}

\noindent \subsubsection*{\textbf{Setup}} To demonstrate the effectiveness of the presented genetic repair algorithm in fixing vulnerable \glspl{sc}, we run the genetic  algorithm on the selected set of \glspl{sc}. We evaluate the effectiveness of the algorithm by measuring the number of vulnerabilities that can be detected and repaired correctly  and the time it takes to generate correct patches of these vulnerable contracts.
Recall that a repair is generated by the algorithm 
when all test cases pass and no targeted vulnerability is found. We call such as a fix as a plausible fix.
Hence, the generated patch might still not be a correct patch.
We then check the correctness of the generated patches by inspecting the semantics of the patches manually. We assert a plausible fix for one vulnerability as correct if it  repaired the vulnerability being detected while the original business logic is not modified and the fix does not introduce new features or vulnerabilities to the code. We use the unguided random search implementation as the baseline to evaluate the effectiveness of the designated guidance in the search process. This is essential since the expensive complete patch quality assessment by the objective functions could result in lowering the efficiency and effectiveness\cite{QiRand2014,Arcuri2011}.

\noindent \subsubsection*{\textbf{Results}} For each of the considered vulnerable contracts, we have run our algorithm five times, each time with a timeout of one hour. We report the average value of the run time and the sum of plausibly successfully repaired vulnerabilities among five runs as the final results. 
Table \ref{table:RQ1-result} shows the summary of the results and the  average run time of the algorithm. 
 The algorithm was able to plausibly repair 26 occurrences of vulnerabilities among the 48 detected vulnerabilities. The average run time of the algorithm over the considered 17 subjects was 25 minutes. We noticed that the main bottleneck of the implementation is due to the test case execution time which often consumes the most computational resources and blocks the synchronization barrier of each iteration of the main loop of the algorithm.
When inspecting the generated patches, we found that our algorithm was able to fix correctly 21 vulnerabilities out of the detected 48 vulnerabilities. With the same timeout, our genetic algorithm was able to plausibly fix 15 more vulnerability than the unguided random search version yielding a 136\% improvement. This clearly shows the guidance in the search process from the genetic algorithm has increased the repair efficiency significantly. 
Moreover, a careful inspection of the results reported in Table \ref{table:RQ1-result}  leads to the following interesting observations.

\begin{RQObservation}{RQ1.1}
    As shown in \autoref{table:RQ1-result} there are four different classes of vulnerabilities that have been considered when evaluating the algorithm, namely, \gls{ed}, \gls{re}, \gls{io}, and \gls{tod}.
    We observed that most of the vulnerabilities of the classes \gls{ed} and \gls{re} have been fixed correctly by the algorithm, where 21 out of the 28 detected \glspl{ed} have been plausibly repaired and 4 out of the 6 detected \glspl{re} have been plausibly repaired.
    On the other hand, the algorithm
    was unable to generate correct patches for any of
    the vulnerabilities of the classes IO and TOD; one plausible patch for IO was generated.
\end{RQObservation}


\begin{RQObservation}{RQ1.2}
The occurrence rates of the vulnerabilities ED, RE, IO, and TOD in the considered vulnerable contracts are as follows:
 ED occurs $58\%$, RE occurs  $13\%$, TOD occurs  $4\%$, and IO occurs  $25\%$.
 We observed that the ED vulnerability is the most frequently occurring class of bugs in the selected vulnerable contracts, where 28 out of the 48 detected bugs are ED bugs.

\end{RQObservation}

\begin{RQObservation}{RQ1.3}
  We observed that 7 out of the considered 17 
  vulnerable contracts have been repaired in less than 
  10 minutes, where most of these contracts contain multiple bugs. This demonstrates clearly the efficiency of the presented parallel genetic repair algorithm in fixing vulnerabilities in a considerably short amount of time.
\end{RQObservation}

\begin{tcolorbox}


Answer to \textit{RQ\theRQ}:
Among the 48 detected vulnerabilities in the 20 vulnerable smart contracts,
the algorithm was able to fix plausibly 26 vulnerabilities,
where 21 of these plausible fixes have been verified to be correctly fixing the vulnerabilities. Notably, our implementation fully repaired 10 of the 20 contracts. 
\end{tcolorbox}

\refstepcounter{RQ}\label{RQ:affect-gas}
\subsection*{RQ\theRQ: Does fixing the vulnerability affect the gas usage?}

\noindent \subsubsection*{\textbf{Setup}} When fixing the detected vulnerabilities, expressions in the vulnerable smart contracts will be modified. However, it is unclear whether plausibly fixing the vulnerabilities would change the average gas consumption of the smart contract.  We therefore perform a comparison on the average gas  consumption between the original vulnerable smart contract and the plausibly patched versions generated from five repeated runs conducted in RQ\ref{RQ-main}. Gas dominance levels between the original contract and the patched versions are computed. We assert the patched version has different average gas  consumption from the original version when they are of different gas dominance levels. 
To calculate the gas dominance level, the gas formula of the original version and the patched versions will be generated, as described in earlier sections.

\noindent \subsubsection*{\textbf{Results}} \autoref{table:RQ2-result} shows the difference in average gas  consumption between the plausible patches and the original version. Subjects for which plausible patches could not be generated within time limit (1 hour) are omitted for consideration of this RQ. To sum up, 6 out of 8 (75\%) of our set of selected subjects have plausible patches with gas formula that are different from the original vulnerable version while half (50\%) of our set of selected subjects have plausible patches having gas dominance levels different from that of the original vulnerable version. This suggests the possibility that fixing vulnerabilities in smart contracts can change the average gas consumption of the original contract. For the subjects with plausible patches amending the average gas consumption, each independent patch generation process has high probability (93.65\% in our experiments) of generating plausible patches of gas dominance levels different from the original version.


\begin{tcolorbox}

Answer to \textit{RQ\theRQ}: 
 In general, when fixing vulnerabilities in a vulnerable smart contract, the gas  should be one of the factors considered in the repair process.

\end{tcolorbox}

\begin{table}
\caption{Results for   RQ\ref{RQ:affect-gas}, showing gas variation between buggy contract and patched versions.}
\scalebox{0.7}{

\begin{tabular}{|l|c|c|c|c|c|}
\hline
\textbf{Name of Contract} & \textbf{\# plausible patches} & \begin{tabular}{c} \textbf{\# patches with diff.} \\ \textbf{gas formula} \\ \textbf{from original} \end{tabular} & \begin{tabular}{c} \textbf{\# patches with diff.} \\ \textbf{gas dominance level} \\ \textbf{from original} \end{tabular} & \begin{tabular}{c} \textbf{Ratio of} \\ \textbf{plausible patches yielding} \\ \textbf{diff. average gas} \end{tabular}\\ \hline
Autonio ICO & 7 & 7&6&85.7\%\\ \hline
Airdrop & 5 & 0&0 & 0\% \\ \hline
Banana Coin & 4 &4& 0& 0\%\\ \hline
XGold Coin & 7&7&0& 0\%\\ \hline
Hodbo Crowdsale & 3&3&3& 100\%\\ \hline
Classy Coin Airdrop & 5&5&5& 100\%\\ \hline
HDL Presale &1&0&0& 0\%\\ \hline
Privatix Presale & 9&8&8& 88.89\%\\ \hline
\end{tabular} 
}
\label{table:RQ2-result}
\end{table}

\refstepcounter{RQ}\label{RQ:plausible-diff-gas}
\subsection*{RQ\theRQ: Can plausible patches vary significantly in average gas consumption?}

\noindent \subsubsection*{\textbf{Setup}} Further, we would like to investigate whether there is a possibility to plausibly fix the vulnerabilities with more than one patch yielding to different average gas consumption across patches. In other words, we intend to understand whether the same bugs can be fixed with patches of different average gas consumption. If the answer is positive, we then justify the need to attempt pursuing a more gas-efficient plausible patch during the search process. We conduct our analysis on the patches generated in RQ\ref{RQ-main} across five repeated runs. We leverage gas dominance levels of patches as a proxy to compare the difference in average gas consumption between patches. We assert a patched version has different average gas consumption from the other when they are of different gas dominance levels. Note that two patched versions have different gas dominance levels when the gas formula of one of the two versions dominates the other. However, to calculate the gas dominance level of generated patches, the gas formulas of the original version and the patched versions need to  be generated first.

\noindent \subsubsection*{\textbf{Results}} \autoref{table:RQ:plausible-diff-gas-result} shows the difference in average gas consumption between the generated plausible patches of selected vulnerable contracts. Subjects for which plausible patches could not be generated within time limit (1 hour) are omitted for consideration of this RQ. For 5 out of 8 subjects (62.5\%), we were able to get a set of plausible patches with more than one corresponding unique gas formulas, indicating the diversity of gas consumption between plausible patches addressing the same set of vulnerabilities.
We noticed that plausible patches have around two gas dominance levels among them, on average, for a given contract.

\begin{RQObservation}{RQ3.1}
For 62.5\% of the considered subjects, there exist plausible patches having different average gas consumption.
\end{RQObservation}


\begin{table}
\caption{Results for   RQ\ref{RQ:plausible-diff-gas}, gas variation among patch candidates is shown.}
\scalebox{0.75}{
\begin{tabular}{|l|c|c|c|c|c|}
\hline
\textbf{Name of Contract} & \textbf{\# plausible patches} & \begin{tabular}{c} \textbf{\# unique gas formula} \\ \textbf{among patches} \end{tabular} & \begin{tabular}{c} \textbf{\# gas dominance levels} \\ \textbf{among patches} \end{tabular}\\ \hline
Autonio ICO & 7  & 7 & 7\\ \hline
Airdrop & 5  & 1 &1\\ \hline
Banana Coin & 4 &2&1 \\ \hline
XGold Coin & 7&5&1 \\ \hline
Hodbo Crowdsale & 3&2&2\\ \hline
Classy Coin Airdrop & 5&1&1\\ \hline
HDL Presale &1&1&1\\ \hline
Privatix Presale & 9&3&3\\ \hline
\end{tabular} 
}
\label{table:RQ:plausible-diff-gas-result}
\end{table}

\begin{tcolorbox}

Answer to \textit{RQ\theRQ}: Different plausible patches can yield various average gas consumption for fixing the same vulnerabilities. We should therefore attempt to guide the search towards more gas-efficient plausible patches besides considering their correctness. 
\end{tcolorbox}

\refstepcounter{RQ}\label{RQ:gas-objective}
\subsection*{RQ\theRQ: How effective is the gas ranking approach at producing low-cost patches?}

\noindent \subsubsection*{\textbf{Setup}} During the patch generation process, we have integrated our proposed gas comparison approach to compare the relative gas usage of generated patches. The relative gas dominance relationship is then used in the genetic patch generation process as a guidance to generate a potentially gas optimized patch. To evaluate systematically the effectiveness of the gas usage objective in producing low-cost patches, we run our repair algorithm on the selected vulnerable smart contracts under two different settings: 
the first setting is when the the gas ranking objective is active (done in RQ\ref{RQ-main})
and the second setting is when the gas ranking objective is deactivated. The first setting is a reuse of patches generated in RQ\ref{RQ-main} while the second setting is additional runs with repeating factor of five and timeout of one hour. 
Later, we run all patches generated in both settings on our generated test cases and collect the average runtime gas usage of each setting. For consistent and fair comparison, we only consider patches fixing all vulnerabilities. Different from RQ\ref{RQ:affect-gas} and RQ\ref{RQ:plausible-diff-gas}, this RQ attempts to expose the change in average gas consumption for the previous usages of the contracts to infer practical gas cost changes.
 
\noindent \subsubsection*{\textbf{Results}} \autoref{table:RQ3-result} shows the summary of average gas usage of patches generated with and without the gas objective being activated. Subjects that plausible patches could not be generated within time limit (1 hour) are omitted for consideration of this RQ. Overall, 6 out of 8 subjects among subjects for which both settings can generate plausible patches (75\%), the gas objective is effective to reduce the average cost of the patches by up to 9.31\% for our subjects. Two subjects (Autonio ICO and Classy Coin Airdrop) do not have varied average gas usage between patches generated in two settings. 
One subject (\emph{MXToken Crowdsale}) does not have plausible patch generated where the gas objective is deactivated in the five repeated runs. In addition, we have also done careful profiling of the algorithm exposing the fact that gas ranking has frequently been the determining factor of patch rankings during the repair process of the selected subjects even though the gas objective is employed as a secondary objective.

\begin{table}
\caption{Results for   RQ\ref{RQ:gas-objective}, showing the average gas usage of the patched versions (on the given tests) in two settings. The percentage shows the improvement on average gas usage when the gas objective is enabled.}
\begin{tabular}{|l|c|c|}
\hline
\textbf{Name of Contract} & \begin{tabular}{c} \textbf{Average gas usage} \\ \textbf{(Gas objective is enabled)} \end{tabular} & \begin{tabular}{c} \textbf{Average gas usage} \\ \textbf{(Gas objective is disabled)}
\end{tabular}
\\ \hline
Autonio ICO & 87092.2 (0\%) & 87092.2  \\ \hline
Airdrop & 73633.4 (0.92\%) & 74316.1 \\ \hline
Banana Coin & 72535.1 (0.01\%) & 72542.3 \\ \hline
XGold Coin & \textbf{46154.6 (6.37\%)} & \textbf{49296.3 } \\ \hline
Hodbo Crowdsale & 38848.3 (0\%) &38848.6 \\ \hline
Classy Coin Airdrop & 72810.5 (0\%) &72810.5 \\ \hline
HDL Presale & 48536.6 (0\%)   &48536.525 \\ \hline
Privatix Presale & \textbf{40323.7 (9.31\%)} & \textbf{44464.46 }\\ \hline
MXToken Crowdsale & 43247.4& No patch generated \\ \hline
\end{tabular}
\label{table:RQ3-result}
\end{table}

\begin{tcolorbox}
Answer to \textit{RQ\theRQ}: 
When enabling the gas objective during repair, we observed that the average
gas consumption of generated patches of four vulnerable contracts
has been reduced comparing to the setting in which the gas objective was disabled.
We observed also that  the average gas of two subjects has been considerably reduced when enabling the gas objective,
where the average gas of the patched version of XGold Coin contract is reduced by $6.37\%$ and the average gas of the patched version of Privatix Presale contract has been reduced by $9.31\%$. This is a considerable reduction as gas costs real money. 

\end{tcolorbox}

\begin{table}[h]
\caption{Results for RQ5, obtained by varying the timeout from 30  minutes to 1 hour}
\scalebox{0.85}{
\begin{tabular}{|l|c|c|}
\hline
\textbf{Name of Contract} & \begin{tabular}{c}\textbf{Vulnerabilities }\\ \textbf{Discovered}\end{tabular} & \begin{tabular}{c}\textbf{Vulnerabilities Plausibly Fixed}\\\textbf{(30mins/1hr timeout)}\end{tabular} \\ \hline
Autonio ICO & \gls{ed}(1) & Same \\ \hline
Airdrop &  \gls{ed}(4) & Same \\ \hline
Banana Coin  & \gls{ed}(1), \gls{re}(1) &Same\\ \hline
XGold Coin & \gls{ed}(2) & Same \\ \hline
Hodbo Crowdsale & \gls{ed}(2) & Same \\ \hline
Lescoin Presale & \gls{ed}(2) & \gls{ed}(0/1) \\ \hline
Classy Coin &  \gls{ed}(1), \gls{re}(1)& Same \\ \hline
Yobcoin Crowdsale & \gls{ed}(2), \gls{re}(1) & \gls{ed}(0/1), \gls{re}(0/1) \\ \hline
Classy Coin Airdrop & \gls{ed}(2) & Same \\ \hline
OKO Token ICO & \gls{ed}(4), \gls{re}(2) & \gls{ed}(0/1), \gls{re}(0/1) \\ \hline
ApplauseCash Crowdsale & \gls{ed}(2), \gls{re}(1) & \gls{ed}(0/1), \gls{re}(0/0) \\ \hline
HDL Presale & \gls{ed}(3) & \gls{ed}(1/3) \\ \hline
Privatix Presale & \gls{ed}(1) & Same \\ \hline
MXToken Crowdsale & \gls{ed}(1) & \gls{ed}(0/1) \\ \hline
dgame & \gls{io}(3), \gls{tod}(1) & Same \\ \hline
Easy Mine ICO & \gls{io}(6), \gls{tod}(1) & Same \\ \hline
Siring Clock Auction  & \gls{io}(3) & Same  \\ \hline
\end{tabular} }
\label{table:RQ4-result}
\end{table}

\refstepcounter{RQ}
\subsection*{RQ\theRQ: How does the time budget impact our effectiveness at fixing bugs?}

\noindent \subsubsection*{\textbf{Setup}} Allocating or estimating a feasible time budget 
to a genetic repair algorithm is an interesting open  problem.
It is crucial as it affects the capability  of the algorithm
in generating plausible patches for a given vulnerable contract.
There are some key factors that should be taken into consideration in order to allocate a feasible time budget to our repair algorithm including: 
(i) the size of the test suite,
(ii) the complexity of the contact (i.e., larger contracts may take longer time to be analyzed than smaller contracts),
and (iii)  the estimated size of the search space
which in turn depends on the number of the mutation operators used by the algorithm and size of the original vulnerable contract.
To address this research question, we choose
to evaluate the algorithm under two different time budgets:
the first is when we set the timeout to 30 minutes
and the second is when we set the timeout to one hour.
The goal is then to measure the number of vulnerable contracts
that have been repaired under the two settings.

\noindent \subsubsection*{\textbf{Results}} Table \ref{table:RQ4-result} shows the results 
of running the algorithm over the selected vulnerable smart contracts 
using two different values 
of the timeout parameter (30 minutes and 1 hour).
As shown in the table, when setting the timeout parameter to 30 minutes the algorithm was able to generate plausible patches for 17 vulnerabilities out of the 48 detected ones,
achieving a success rate of $35.4\%$.
On the other hand, when setting the timeout parameter
to 1 hour the algorithm was able to 
generate plausible patches for 26 vulnerabilities, 
achieving a success rate of $54.2\%$.
While the amount of improvement on the repair
rate looks somewhat small, it is very crucial
as it shows that some vulnerabilities
can be only  repaired when increasing
the timeout to 1 hour.
This clearly demonstrates
the impact of the timeout parameter
on the effectiveness of the algorithm. 
However, since every detected vulnerability in a given vulnerable smart contract needs to be repaired and the fact that
the size of the search space can be extremely large, the time budget allocated to the algorithm can play a key role in the successful termination of the algorithm. When we increase the time budget of
 the algorithm, we increase the size of the explored search space which in turn increases the probability of generating plausible patches.


\begin{tcolorbox}

Answer to \textit{RQ\theRQ}:
When we increase the timeout parameter of the algorithm
from 30 minutes to 1 hour we observe that the vulnerability repair rate of the algorithm has been increased from $35.4\%$ to $54.2\%$,
where the genetic algorithm was able to repair 9 extra vulnerabilities. 
This demonstrates clearly the importance
of allocating a substantial time budget (at least one hour) to the algorithm
when repairing vulnerable smart contracts.

\end{tcolorbox}







\subsection{Threats to Validity}

Finally, we discuss the threats to validity of our experimental results.

\paragraph*{Internal validity}
Threats to internal validity are related to the representative nature of our conclusions and summaries made based on our experiment results. In our experimental study, we have conducted our experiment on a sample dataset to  evaluate our approach. The size of the dataset is however limited since this is the first automated smart contract repair work, and therefore, there is no consolidated dataset for use like Detect4J\cite{defects4j} for Java. We are aware that our approach employs biased random search techniques, and for this reason each experiment was repeated five times.
We admit that the presented results are potentially skewed even though we have conducted our experiments with a replication factor of five times for each setup.

\paragraph*{External validity}
External validity treats are related to the ability to generalize our findings. We have only evaluated our work on four known vulnerability types. While our approach is vulnerability-agnostic, the efficacy in terms of fixing other vulnerabilities remains unknown. On the other hand, we have conducted our experiments on real-world subjects as an attempt to investigate the performance of approach. This does not guarantee that similar efficacy will be exhibited for arbitrary vulnerable smart contracts. 

\section{Related Work}

We discuss the related literature on automated program repair, smart contract analysis, and gas usage calculation of smart contracts. 

\subsection{Automated Program Repair}

Automated program repair \cite{GOUES19} has been the subject of considerable recent attention in the software engineering research community. Commonly, they attempt to automate the process of fixing the bugs exposed by failing test cases, and these techniques are collectively called \emph{test-based repair techniques}. The patch that can fix all the given tests is called a \emph{plausible patch}.
Several test-based program repair approaches have been developed. These approaches can mainly be classified into search-based and semantics-based approaches.

Search-based approaches developed by \cite{GouesNFW12}, and \cite{Martinez2015} show promising results towards the automation of bug fixing. The key idea of their approaches is to use failing test cases to identify bugs and then apply mutations to the source code until the program passes all failing test cases, while continuing to pass previously passing tests. Genetic programming \cite{GouesNFW12} as well as random search \cite{QiRand2014} have been used as search techniques for finding a plausible patch, a patch passing given test-cases. GenProg \cite{GouesNFW12} is one of the early works among search-based repair techniques.

Semantic analysis techniques like SemFix \cite{Nguyen2013}, Nopol \cite{Xuan2017}, DirectFix \cite{Mechtaev2015}, SPR \cite{Long2015}, Angelix \cite{Mechtaev2016} and JFIX \cite{Le2017} split patch generation into two steps. First, they infer a desired specification (or a repair constraint) for the buggy program statements, which is often accomplished via symbolic execution of the given tests. Second, they synthesize a patch for these statements based on the inferred specification, using program synthesis techniques. These works view program repair as a specification inference problem, as opposed to searching among candidate patches. These approaches can be combined with search: we explore patches by considering insert/delete/replace of statements, while the semantic analysis can help synthesize expressions to be inserted in the statement replacements \cite{Yi2017}.

Apart from automated program repair approaches driven by functionality (often exposed by a test-suite), some other studies e.g. Caramel \cite{caramel} attempts to automatically fix non-functionality bugs (such as performance bugs). Such bugs can be fixed by inserting an early termination statement inside loops. The generated patch can potentially reduce the run time of the program.

Our smart contract repair problem (defined in Problem \ref{problem:ascr}) is similar to the test-based program repair problem. We also leverage the test cases to examine functional correctness of patches. However, since vulnerabilities in smart contracts have been raising serious financial losses, our  patches need to not only be test-adequate but also secure and gas-aware.

\subsection{ Testing and Analysis of Smart Contracts}

Analysis  of   smart  contracts  is  a  popular  topic  that  has  received  a  lot  of  attention recently, with numerous tools being developed based on fuzz testing, symbolic execution and constraint solving. \cite{Luu2016,Grossman2017,KalraGDS18,Jiang2018,Amani2018,Meyden2019}. Oyente \cite{Luu2016} is one of the earlier works on symbolic analysis of smart contracts. The work in \cite{Amani2018}  translates smart contract source code to
Isabelle/HOL in order to validate smart contracts.
The authors use the symbolic security analyzer Oyente \cite{Luu2016}  to detect vulnerabilities in smart contracts. 
The tool ContractFuzzer \cite{Jiang2018} 
 uses fuzz testing to detect security vulnerabilities in smart contracts.  Recently, van der Meyden \cite{Meyden2019}  conducted
a formal analysis of an abstract  model  of  smart contract code (atomic swap smart contracts)  using the epistemic MCK model checking tool \cite{Meyden2004}.
He showed how to  automatically  verify  that  a  concrete  implementation  of  atomic  swap  satisfies  its specification using epistemic-temporal logic model checking.

There is also a considerable amount of work on the mutation testing of smart contracts  \cite{HaoranWu2019,HonigEH19,Fu2019}. Mutation testing \cite{PapadakisK00TH19,Offutt2001}  is a technique for evaluating the quality of a set of test cases (i.e., a test suite). It works by introducing faults into a system via source code mutation and then analyzing the ability of some developed test suite to detect these faults. The work in \cite{HaoranWu2019} has implemented certain mutation operators and tested them on four DAPPS (decentralized applications on blockchain). However, their  approach does not take into consideration the access control faults and the gas usage of the mutated contracts. 
The work in \cite{HonigEH19} developed a mutation testing framework
for smart contracts that considers the access
control faults, but it does not consider gas.
The work in \cite{Fu2019} introduced  a  smart contract mutation approach, but for testing implementations of the Ethereum Virtual Machine (EVM) implementations and not smart contracts. There are
two available GitHub repositories with related tools on mutation testing of smart contracts:  (1) Eth-mutants\footnote{\url{https://github.com/federicobond/eth-mutants}} which implements
just one mutation operator and (2) UniversalMutator
which describes a generic mutation tool \cite{GroceHMSZ18} with set of operators for
Solidity.  

None of the past works on testing and analysis, focus on automated repair of smart contracts. These works are focused on finding bugs in smart contracts, and not on fixing bugs. Ours is the first proposed approach and tool for smart contract repair.

\subsection{Gas usage Calculation of Smart Contracts}

The work in \cite{MarescottiBHAS18}  presented techniques for calculating the worst case gas usage of smart contracts. 
 Their approach is based on symbolically enumerating all execution paths and unwinding loops up to a certain limit. The authors infer the maximal number of iterations 
 for loops and generates accurate gas bounds.
 Knowing the worst case gas usage bound for smart contracts 
 can be extremely useful   as it provides the smart contract users important information about the maximum amount of gas they need to pay before sending out their transactions to the blockchain networks.  The work in \cite{Signer2018} provides a
graphical user interface that depicts gas usage information 
(e.g. best and worst case gas usage, and the gas usage of different parts of the code) which helps the developers to optimize the gas usage of their smart contracts.

The past works on gas usage calculation, while relevant to our works, are not directly usable in our repair method. For fast gas usage comparison among patch candidates, we have thus defined and used the notion of gas dominance.

\section{Discussion}

In this paper, we have presented the first work on automatically repairing smart contracts. Our 
repair method is gas-aware.
The repair algorithm is search-based, and it breaks up the huge search space of candidate patches down
into smaller mutually-exclusive  spaces that can be processed independently.
The repair technique considers gas usage of vulnerable contracts
when generating patches for detected vulnerabilities. 
Our experiments demonstrated that our method can  handle  real-world contracts
and generate repairs in a short time (less than 1 hour) while taking into consideration the gas consumption of the generated repairs.  

Since the owners of smart contracts are unknown, we could not reach out to them in advance, prior to publication. Nevertheless, we hope that our work will spur greater interest in automatically fixing smart contracts via a variety testing, analysis, validation and synthesis methods. We have made our smart contract repair tool and dataset available in GitHub from the following site.
\begin{center}
  {\tt \url{https://SCRepair-APR.github.io}}  
\end{center}

\section*{Acknowledgments}
This work was partially
supported by the National Satellite of Excellence in Trustworthy
Software Systems, funded by National Research Foundation (NRF) Singapore under National Cybersecurity R\&D (NCR) programme, and by a Singapore Ministry of Education (MOE) Academic Research Fund (AcRF) Tier 1 grant (17-C220-SMU-008).

\bibliographystyle{ACM-Reference-Format}
\bibliography{SCRepair}

\end{document}